\documentclass[a4paper,11pt]{article}
\usepackage{jheppub}

\usepackage{graphicx}
\usepackage{amsmath}
\usepackage{amssymb}
\usepackage{hyperref}
\usepackage{amsthm}
\newtheorem{theorem}{Theorem}
\newtheorem{lemma}[theorem]{Lemma}
\newtheorem{corollary}[theorem]{Corollary}
\theoremstyle{definition}
\newtheorem{definition}{Definition}[section]

\usepackage{enumerate}
\usepackage[utf8]{inputenc}
\usepackage{amssymb,amsbsy,mathtools,amsmath}
\usepackage{braket}
\usepackage{graphicx}
\usepackage{xcolor}
\usepackage{hyperref}
\usepackage{natbib}
\usepackage{bm}
\usepackage{soul}

\usepackage[nottoc]{tocbibind}

\usepackage{bbm}
\usepackage{amsthm}
\usepackage{amssymb}
\usepackage{amsmath}
\usepackage{tikz}
\usepackage{braket}
\usepackage{quantikz}
\usepackage{graphicx}
\usepackage{caption}
\newcommand{\mA}{\mathcal{A}}

\newcommand{\nn}{\nonumber}

\newcommand{\lb}{\left(}
\newcommand{\rb}{\right)}

\usepackage{subcaption}

\newcommand{\mO}{\mathcal{O}}

\newcommand{\ep}{\epsilon}

\newcommand{\p}{\partial}
\newcommand{\tr}{\text{tr}}

\newcommand{\ad}{\text{ad}}
\newcommand{\Ad}{\text{Ad}}

\begin{document}

\title{Local approximations of global Hamiltonian from inclusion of algebras}


\author{Yidong Chen$^1$,}
\affiliation{$^1$Department of Mathematics, University of Illinois at Urbana-Champaign, Illinois, IL 61801, USA}
\author{Nima Lashkari$^2$,}
\author{Kwing Lam Leung$^2$}
\affiliation{$^2$Department of Physics and Astronomy, Purdue University, West Lafayette, IN 47907, USA}
\emailAdd{yidongc2@illinois.edu}
\emailAdd{nima@purdue.edu}
\emailAdd{leung60@purdue.edu}


\abstract{We write down the global Hamiltonian of conformal field theory (CFT) in finite volume in terms of the modular Hamiltonian of the vacuum reduced to a local ball-shaped region, and use it to propose local approximations to the global Minkowski Hamiltonian in quantum field theory (QFT). The proposed Hamiltonians are motivated by the operator-algebraic property of nuclearity. They are constructed from the characteristic functions of inclusion of algebras and can be viewed as regulators of the modular Hamiltonian of local algebras of QFT.}


\maketitle

\section{Introduction}
In most applications of quantum information and entanglement theory to QFTs, we use an ultra-violet (UV) regulated Euclidean path-integral to associate density matrices $\rho_A$ to causally complete regions of spacetime $A$. For the Rindler wedge $W:x^1>|x^0|$, up to UV cut-off dependent terms, the density matrix is the exponentiated boost generator $B_1$ in the $x^1$-direction restricted to $W$. The full modular operator of $W$ is $\Delta_W=\rho_W\otimes \rho_{W'}^{-1}\sim e^{-2\pi B_1}$ where $W'$ is the causal complement of $W$ \footnote{More generally, in any geometry with a bifurcate Killing horizon, the modular Hamiltonian of half-space $W$ is given by the generator of the Killing flow $K_W$.}. As we take the UV cut-off away, the non-universal boundary terms in the modular Hamiltonian $K_W=-\log \Delta_W$ drop, and we are left with boost which has the entire real line in its spectrum and no normalizable eigenvector\footnote{In operator-algebraic language, the local algebras of $W$ and $W'$ in QFT do not split, and do not admit density matrices.}. In the presence of density matrices, for real $t$, the trace of modular flow  $\tr(\Delta^{it})=|\tr(\rho^{it})|^2$ can be interpreted as a modular spectral form factor. For local algebras of QFT, this trace diverges and needs to be regulated.

The divergences above are not specific to half-space regions or the vacuum state. It has been argued in \cite{fredenhagen1985modular} that for any state of QFT and the causal development of a ball $B$, the modular Hamiltonian always has the entire real line in the spectrum, which is the defining property of type III$_1$ von Neumann algebras. It is believed that the local algebra of observables of any QFT, in any state, and for any causality-complete local region is the unique type III$_1$ {\it hyperfinite} factor \cite{buchholz1987universal}. This uniqueness implies that a unitary change of basis takes local operators of one-dimensional free bosons on an interval to those of the causal development of a ball in any QFT in any spacetime dimension. In other words, to describe the physical properties of  QFT, we need a more refined operator-algebraic structure than just one abstract algebra. The key intuition is that the difference between different QFTs is encoded in the inclusion of the algebra of regions. 

The proof of the operator algebraic property of hyperfiniteness already hints at the intuition above \cite{buchholz1987universal}. It relies on a class of requirements called {\it nuclearity} that are operator algebraic characterizations of the physical requirement that the growth of the density of states at high energies is nicely bounded \cite{buchholz1986causal}\footnote{The nuclearity assumptions imply the split property, which in turn implies hyperfiniteness. Note that the split property does not imply nuclearity.}. They are manifestations of locality, and some authors include nuclearity in the axioms of local QFT \cite{hollands2018entanglement}. To understand these operator algebraic requirements, we recall that the GNS representation map sends operators to vectors
\begin{eqnarray}
    \Theta:\mA\to \mA\ket{\Omega}, \qquad \Theta(a)=a\ket{\Omega}\ .
\end{eqnarray}
Since local algebras (of causally complete regions of finite volume) of QFT are infinite-dimensional, the image of the map is infinite-dimensional, as well \footnote{In fact, in QFT, it follows from the Reeh-Schlieder property that $\overline{\mA_A\ket{\Omega}}=\mathcal{H}$ for all local algebras $\mA$.}. This means that by acting locally on the vacuum state, we have access to create an infinite-dimensional Hilbert space of excitations. A common regulator used in QFT is Euclidean time evolution:
\begin{eqnarray}\label{BWregulator}
    \Theta^{BW}_\ep(a)=e^{-\ep H}a\ket{\Omega}\ .
\end{eqnarray}
With the regulator, the effective dimensionality of the Hilbert space we explore becomes finite \cite{haag2012local}. To measure the effective dimensionality of the map's range, we use nuclear norms, which are natural generalizations of the trace distance norm from Hilbert spaces to Banach spaces; see Appendix \ref{subsection:defnuclear} for a definition \footnote{Note that the domain of $\Theta^{BW}_\ep$ map is the algebra itself, which is a non-commutative $L^\infty$ space, and not an $L^2$-space (Hilbert space) \cite{furuya2023monotonic}.}. The nuclear norm of $\|\Theta_\ep\|_1$ is finite but diverges as $\ep\to 0$. Intuitively, we expect that as we decrease $\ep$, the states in the range of (\ref{BWregulator}) can have larger energies. As we take $\ep\to 0$, the growth of the nuclear norm of $\Theta_\ep$ is controlled by the growth of the density of states at high energies. The Buchholz-Wichmann (BW)-nuclearity \cite{buchholz1986causal} is the requirement that there exists a constant $c$ and $n(R)>0$ such that for the local algebra of a ball of radius $R$ we have
\begin{eqnarray}
    \|\Theta^{BW}_{\ep,B(R)}\|_1\leq e^{\lb\frac{c}{\ep}\rb^{n(R)}}\ .
\end{eqnarray}
This condition is tied to the familiar statement that in physical theories, the microcanonical entropy grows as a power law in energy \cite{narnhofer1994entropy}. The downside of this condition is that it requires access to the global Hamiltonian as a regulator. It is desirable to have an operator algebraic characterization of this condition that is solely in terms of the local data. Modular nuclearity and $L^2$-nuclearity are such local conditions that, instead of the global Hamiltonian, use the inclusion of algebras as regulators \cite{buchholz1990nuclear,buchholz1990nuclear2,buchholz2007nuclearity}. For $B_s$ a linear scaling of $B$ with $s>0$, the inclusion $B_s\subset B_0$ for small $s$ can be viewed as a Lorentz-invariant regulator for $B_s$. This operator algebraic construction has been used as a regulator for entanglement measures in QFT  \cite{narnhofer2002entanglement,hollands2018entanglement,casini2015mutual}.
In modular nuclearity, we use the following local regulator
\begin{eqnarray}
&&\Theta^M_{\ep,\alpha,B_s\subset B}:\mA_{B_s}\to \Delta_B^\alpha \mA_{B_s}\ket{\Omega}
\end{eqnarray}
The domain and range of the maps above are all non-commutative $L_p$-spaces. 
It is more convenient to work with Hilbert space (non-commutative $L^2$-space) so that the nuclear norm becomes trace distance. This motivates the condition of $L^2$-nuclearity, defined as the nuclearity of the map 
\begin{eqnarray}
    \Theta^{L2}_{\ep,B}:\Delta_{B_s}^{1/4}\mA_{B_s}\ket{\Omega}\to \Delta_B^{1/4}\mA_{B_s}\ket{\Omega}
\end{eqnarray}
Then, the nuclear norm of this map is 
\begin{eqnarray}\label{L2map}
    \|\Theta^{L2}_{\ep,B}\|_1=\tr(|\Delta_B^{1/4}\Delta_{B_s}^{-1/4}|)
\end{eqnarray}
The assumption of $L^2$-nuclearity in two-dimensional CFTs was explored in \cite{buchholz2007nuclearity}. We use their work extensively here \footnote{For a recent discussion of applications of $L^2$-nuclearity to holography see \cite{Ceyhan:2025qrj}.}. They showed that $L^2$-nuclearity implies modular nuclearity, which implies BW-nuclearity. This suggests that the UV behavior of the global Hamiltonian is reflected in the local algebras. The interpolation theory of non-commutative $L_p$-spaces suggests a generalization of (\ref{L2map}) to \cite{furuya2023monotonic}
\begin{eqnarray}
    \lb\Delta_B^\alpha \Delta_{B_s}^{-2\alpha}\Delta_B^\alpha\rb^{1/(2\alpha)}
\end{eqnarray}
that we study in more depth.

In this work, we argue that the local data of the inclusion of balls $B(e^{-2\pi s}R)\subset B(R)$ in the limit $s\to 0$ can be used to obtain a local approximation to the global Hamiltonian, regulate the modular Hamiltonian, and other information-theoretic measures of entanglement. 
Specifically, for any inclusion of ball-shaped regions of radii  $e^{-2\pi s}R$ and $R$ and $s>0$, i.e. $B(e^{-2\pi s}R)\subset B(R)$, and $\Im(z)\in (-1/2,0)$ we study the {\it characteristic function} of the inclusion (see Appendix \ref{subsection:characteristicfunctions} for a review of characteristic functions of inclusion of von Neumann algebras)
\begin{eqnarray}
    T_R(z)=\Delta_{B(R)}^{iz} \Delta_{B(e^{-2\pi s}R)}^{-iz}
\end{eqnarray}
and the operator
\begin{eqnarray}\label{proposal}
    \tilde{H}_R(z)=&&\frac{{-1}}{2\sin(2\pi i z)}\p_s\log(T_R(z)T_R(-z)^{-1})_{s=0}
    \nn\\
    =&&\frac{{-1}}{2\sin(2\pi i z)}\p_s\lb\Delta_{B(R)}^{i z} \Delta_{B(e^{-2\pi s}R)}^{-2iz} \Delta_{B(R)}^{iz}\rb_{s=0}
\end{eqnarray}
Note that this definition is even in $z$, and for both real and imaginary values of $z$, $\tilde{H}_R(z)$ is self-adjoint. For $z=-i\alpha$ purely imaginary, $\tilde{H}_R(z)\geq 0$ \footnote{This follows from the fact that for $z=-i\alpha$ with $\alpha\in (0,1/2)$ for some positive coefficient $C$ and positive operator $Q$ we have $\tilde{H}_R(z)=C Q(\p_s\Delta_{B(e^{-2\pi s}R)}^{2\alpha})Q$ and we have used $\p_s\Delta^{-1}=-\Delta^{-1}(\p_s\Delta)\Delta^{-1}$. Since $\p_s=-2\pi R\p_R$, it follows from the monotonicity of the modular operator that $\p_s\Delta_s>0$, hence, $\tilde{H}_R(z)>0$.}.

$L^2$-nuclearity is the statement that $\text{tr}(e^{-\beta\tilde{H}_R(-i/4)})<\infty$ for all positive $\beta$ \footnote{Note that this trace class condition does not imply that the spectrum of $\tilde{H}_R(z)$ is discrete.}. 
In this work, we establish the following results:
\begin{enumerate}
    \item In vacuum CFT on Lorentzian cylinder of circumference $2\pi$, and the modular Hamiltonian $K_u\equiv K_{B_u}$ of ball-shaped regions of radius $R=e^{-2\pi u}$, our self-adjoint operator 
    $\tilde{H}_u(z)=\p_u K_u$ and independent of $z$. In CFT, we have $\tilde{H}_{u}= \Delta_F^{-iu}K_H \Delta_F^{iu}$ where $K_H = 2\pi H$, which is basically the dilation of the global Hamiltonian of the Lorentzian cylinder.
    To isolate the global Hamiltonian of the cylinder from the local state, we need to undo this dilatation:
    \begin{eqnarray}
        2\pi H=\cosh(2\pi u)\tilde{H}_u-\sinh(2\pi u)K_u\ .
    \end{eqnarray}
    It follows from the representation theory of the conformal group that $H$ has a discrete spectrum; see Appendix \ref{subsubsection:discrete}.
    \item In a QFT, we have
    \begin{eqnarray}\label{Htilde}
        &&\tilde{H}_R(z)=\mathcal{E}_z\lb R\p_R K_{B(R)}\rb
    \end{eqnarray}
    where $\mathcal{E}_z$ is a coarse-graining map (super-operator) that smoothes operators by averaging them over Lorentzian/Euclidean modular time
    \begin{eqnarray}\label{Ez}
        \mathcal{E}_z(X)=\frac{1}{2\sinh(2\pi z)}\int_{-z}^zdt\: e^{i t K}X e^{-i t K}\ .
    \end{eqnarray}
To understand why we call the map $\mathcal{E}_z$ coarse-graining, assume that there are eigenoperators 
\begin{eqnarray}\label{eigenop}
    e^{it K}X_\omega e^{-i t K}=e^{it \omega}X_\omega\ .
\end{eqnarray}
Then, the map $\mathcal{E}_z(X_\omega)$ acts a low-pass filter
\begin{eqnarray}\label{filter}
    \mathcal{E}_z(X_\omega)=\frac{1}{2\sinh(2\pi z)}\int_{-z}^z dz e^{i\omega z}=\frac{\sin(\omega z)}{\omega \sinh(2\pi z)}X_\omega
\end{eqnarray}
because $\sin(\omega z)/\omega\to 1$ as $\omega\to 0$, and decays like $\frac{1}{\omega}$ for large $\omega$. An alternative expression is
\begin{eqnarray}
    \tilde{H}_R(z)={\frac{-\pi}{\sin(2\pi z)}}\{\Delta^{-z},R\p_R(\Delta^z)\}
\end{eqnarray}
where $\{A,B\}=AB+BA$ is the anti-commutator.

  \item  We use the expression in Lemma \ref{Htilde} to propose local approximations of the global Minkowski Hamiltonian in QFTs in the vicinity of a UV fixed point:
    \begin{eqnarray}\label{localapprox}
        2\pi H_z=\frac{1}{R}\lb \tilde{H}_R(z)+ K_{B(R)}\rb\ .
    \end{eqnarray}
\end{enumerate}
By construction, the operator $\tilde{H}_R(z)$ is positive, however, once we add the extra term in (\ref{localapprox}) the positivity is no longer guaranteed. The approximations $H_z$ for all $z$ are expected to tend to the global Minkowski Hamiltonian at high enough energies. Different $H_z$ differ by the part of the infrared physics they preserve. More generally, we propose 
\begin{eqnarray}\label{averagemeasure}
&& 2\pi H_p=\frac{1}{R}\lb \tilde{H}_p+K_{B(R)}\rb\nn\\
&&    \tilde{H}_p=\int_\mathcal{C} dz \: p(z) \mathcal{E}_z(R\p_R K_{B(R)})
\end{eqnarray}
for some probability distribution $p(z)$ supported on some contour $\mathcal{E}$ in the complex $z$-plane.  The basic idea is that one needs to set up an optimization problem to fix $p(z)$ depending on the problem of interest, and what aspects of the infrared physics of the global Hamiltonian we would like to reproduce. For instance, one might require that the reconstruction is such that it matches the first-order correction in conformal perturbation theory, or one might require that in the vicinity of the entangling surfaces, the proposed approximation becomes exact. Note that the different choices of $p(z)$ correspond to different filters in the sense of (\ref{eigenop}). We choose $p(z)$ to be a probability distribution so that the positivity of $\tilde{H}_R(z)$ implies the positivity of $\tilde{H}_p$. The normalization $\int_{\mathcal{C}} dz p(z)=1$ is required by the desire that the answer matches the CFT Hamiltonian at the fixed point.

Our work has implications for quantum chaos. To see this, consider a thermal density matrix (Gibbs state): $\rho\sim e^{-\beta H}$. Then, $\tr(\Delta^{i T})$ for real $T$ computes the spectral form factor: $\tr(\Delta^{i T})\sim |\tr(e^{-i\beta H T})|^2=|Z(\beta i T)|^2$. It is a key conjecture in quantum chaos that for a chaotic Hamiltonian, the coarse-grained spectral form factor shows a universal decay, ramp, and plateau \footnote{This can be viewed either as an infinite temperature, or microcanonical spectral form factor. A form of coarse-graining often discussed in the literature is averaging over time, which resembles our coarse-graining in (\ref{Ez}).}. In Section \ref{sec:regulate}, we will see that our operators $\tilde{H}_z$ can be used to regulate the modular spectral form factor.

In holography, thermal states of the boundary are dual to eternal black holes in Anti-de Sitter (AdS) space. Saad-Shenker-Stanford suggested that the gravity partition function of the double-cone wormhole one obtains by periodically identifying real time $t\sim t+T$ in the eternal AdS black hole computes the spectral form factor \cite{saad2018semiclassical}. However, the real-time identification induces a singularity at the bifurcate Killing horizon. In \cite{saad2018semiclassical}, the authors suggested regulating this trace using an $i\ep$-prescription. This regulator, in essence, coarse-grains the modular Hamiltonian using a non-self-adjoint deformation $K_W \to K_W +i \epsilon H$, where $H$ is the global Hamiltonian. They argued that the spectrum of this modified modular Hamiltonian captures modular quasi-normal modes \cite{chen2024comments}. We discuss the connection between this bulk regulator and our operator $\tilde{H}_R(z)$ in Section \ref{sec:regulate}.



In Section \ref{sec:Rindler}, we study the characteristic function of the inclusion of wedges in QFT and show that our proposed Hamiltonian in (\ref{proposal}) matches the global Hamiltonian.
Section \ref{sec:CFT} derives the local reconstruction of the global Hamiltonian of the Lorentzian cylinder in vacuum CFTs from the characteristic function of the inclusion of ball-shaped regions. In Section \ref{sec:QFT}, we extend our proposal to QFT and study it in perturbation theory. 

\section{Regulating Modular Hamiltonian}\label{sec:regulate}

To study the chaotic properties of modular evolution, we would like to compute and regulate $\tr(\Delta^{iT})=\tr(\rho^{iT})\tr(\rho^{-iT})$. To build intuition, let us start by regulating $\tr(\Delta^\alpha)$ with some real $\alpha\in(0,1/2)$. This cuts off the large modular-energy behavior of $K$, but in our cases of interest, our local algebras are type III$_1$, and as we discussed before, the spectrum of $K$ is the entire real line. The modular Hamiltonian has no normalizable eigenvectors, and the modular spectral form factor diverges. Euclidean path-integrals help us understand the physical origin of this divergence. For Rindler wedges $W$, the operator $\Delta^\alpha$ is a Euclidean rotation by angle $\alpha$ around the entangling surface. The trace $\tr(\Delta^\alpha)$ diverges because there are matrix elements $\braket{a\Omega|\Delta^\alpha |a\Omega}$ that become arbitrarily large \footnote{The operator $a$ have support on large negative modular frequencies. They can be regulated by filters of the form (\ref{filter}) that suppress large $\omega$.}.
Our intuition from the Euclidean path integral tells us that we can regulate this divergence by separating the operator insertions, by acting on the bra and ket by Euclidean translations in space or time in opposite directions:
\begin{eqnarray}
&&\ket{a\Omega}\to \ket{e^{-\ep P_1}a\Omega},\qquad \bra{a\Omega}\to \bra{e^{\ep P_1} a\Omega}\nn\\
&&\ket{a\Omega}\to \ket{e^{-\ep H}a\Omega},\qquad \bra{a\Omega}\to \bra{e^{\ep H} a\Omega}\ .
\end{eqnarray}
The first case is equivalent to deforming the modular Hamiltonian by
\begin{eqnarray}
        &&\Delta^\alpha\to e^{\ep P_1}\Delta^\alpha e^{-\ep P_1}=e^{-\alpha (K+\ep[P_1,K])}+O(\ep^2)=e^{-\alpha(K-i\ep H)}+O(\ep^2)
\end{eqnarray}
where we have used the Poincaré algebra\footnote{We are using the  convention$(-1,+1,\cdots, +1)$ for the signature of the Minkowski spacetime.}
\begin{eqnarray}
     [K_1,P_1]=i H,\qquad [K_1,H]=i P_1\ .
\end{eqnarray}
This is precisely the analytic continuation suggested by Saad-Shenker-Stanford in \cite{saad2018semiclassical}. In \cite{chen2024comments}, it was argued that the regulated spectral form factor computes the spectrum of quasi-normal modes:
\begin{eqnarray}
    \tr(e^{-i(K-i\ep H)T})=\prod_{\omega_{QNM}}\frac{1+e^{-i\omega_f T}}{1-e^{-i\omega_b T}}
\end{eqnarray}
where $\omega_b$ and $\omega_f$ are the bosonic and fermionic quasi-normal modes, respectively.

In comparing to the double cone regulator above, our regulators that are based on the inclusion of algebras, regulate correlators by separating them in Euclidean time (instead of space), as in the case of BW-nuclearity:
\begin{eqnarray}
    f_{aa}^{\ep}(\alpha)&&=\braket{\Delta^{1/4}a\Omega|e^{\ep H}\Delta^{\alpha} e^{-\ep H}\Delta^{1/4}a\Omega}\nn\\
    &&=\braket{\Delta^{1/4}e^{-i\ep P} a\Omega|\Delta^\alpha\Delta^{1/4}e^{i\ep P}a\Omega}
\end{eqnarray}
where we have used the relations 
\begin{eqnarray}\label{PoincareQEC}
    \Delta^{1/4}e^{-i\ep P}=e^{-\ep H}\Delta^{1/4}\nn\\
    \Delta^{1/4}e^{-i\ep H}=e^{-\ep P}\Delta^{1/4}\ .
\end{eqnarray}
We can write the first relation above as
\begin{eqnarray}
   &&T_\ep=\Delta^{1/4}_W\Delta_{W(\ep)}^{-1/4}= e^{-\ep (H-i P)}\nn\\
   &&\mA_{W(\ep)}=e^{-i\ep P}\mA_W e^{i\ep P}\ .
\end{eqnarray}
Here, $\mA_{W(\ep)}\subset \mA_W$ is a split inclusion of Rindler wedges. Note that in this case, the global Hamiltonian is simply
\begin{eqnarray}
    H=\frac{-1}{2\ep}\log\lb T_\ep T_\ep^\dagger\rb=-\frac{1}{2}\p_\ep\lb T_\ep T_\ep^\dagger\rb_{\ep=0}\ .
\end{eqnarray}
The second relation in (\ref{PoincareQEC}) suggests that the double cone regulator of \cite{chen2024comments} that encodes quasi-normal modes, involves a translation of the region in time. For the algebra of wedges, this does not result in inclusions. However, for the algebra of forward light-cone, time translation gives half-sided inclusions. 

In the discussion above, we assumed the existence of a Poincaré algebra of $K$, $P$, and $H$. A natural question to ask is whether this can be generalized to regulate the modular Hamiltonian of type III$_1$ algebras of a more general quantum system. Here, the key intuition comes from non-commutative ergodic theory and quantum Anosov systems. Quantum Anosov systems in modular dynamics admit a pair of future and past modular subalgebras \cite{ouseph2024local}. It was shown in \cite{ouseph2024local} that in such systems, by considering the inclusions $A(s)\subset A(0)$ in the limit $s\to 0$, we can find an approximate local two-dimensional Poincare algebra required for our discussion above. 
This includes examples of General Free Fields (GFF) dual to the double cone geometry studied in \cite{chen2024comments}.

\section{QFT Hamiltonian from Half-space}\label{sec:Rindler}

Let us start with the vacuum state of QFT in Minkowski spacetime $\mathbb{R}^{d,1}$ with flat metric $ds^2=dx^+dx^-+dx^idx^i$ and null coordinates $x^\pm=x^1\pm x^0$. A pair $(z^+,z^-)$ defines a right wedge and a left wedge that is the causal complement
\begin{eqnarray}
 &&W(z^+,z^-)=(x^+>z^+,x^->z^-,x^i)\nn\\
 &&W'(z^+,z^-)=(x^+<z^+,x^-<z^-,x^i)\nn\\
 &&W(z_0)\equiv W(z_0,z_0)\ .
\end{eqnarray}
In flat space, for Rindler wedges $W(x_0)$ we have
\begin{eqnarray}
    K_{W(x_0)}=2\pi\int dx_\perp\int dx\: (x-x_0)T_{00}(x)
\end{eqnarray}
It is not surprising that one can construct the global Minkowski Hamiltonian using only the data of the Rindler wedge: 
\begin{eqnarray}
    -\frac{1}{2\pi}\p_{x_0}K_{W(x_0)}=H\ .
\end{eqnarray}
Note that as we increase $x_0$, the region becomes smaller, and it follows from the monotonicity of $K_W$ that $H>0$.
\begin{theorem}\label{thmRinlder}
    For the inclusion of wedges $W(a)\subset W(0)$, for all complex $z$ with $\Im(z)\in (-1/2,0)$ we have
    \begin{eqnarray}
    &&T_a(z)\equiv \Delta_{W(0)}^{iz}\Delta_{W(a)}^{-iz}=e^{i a(\sinh(2\pi z) P_0+(\cosh(2\pi z)-1)P_1)}\nn\\
    \label{equation:keyidentity}
    &&T_a(z)T_a(-z)^{-1}=\Delta_{W(0)}^{i z}\Delta_{W(a)}^{-2i z}\Delta_{W(0)}^{iz}=e^{2i a H\sinh(2\pi z)}\label{T_a2}\\
    &&H_a(z)\equiv \frac{1}{2\sin(2\pi iz)}\p_a \log (T_a(z)T_a(-z)^{-1})=H\ .
\end{eqnarray}
\end{theorem}
The proof follows from direct calculation and the Poincaré algebra; see \cite{buchholz2007nuclearity}. It follows that 
\begin{eqnarray}
    &&T_a(-i/4)=\Delta_{W(0)}^{1/4}\Delta_{W(a)}^{-1/4}=e^{-a P_0}e^{-i a P_1}\nn\\
    &&T_a(-i/2)=J_BJ_A=e^{-2ia P_1}
\end{eqnarray}
where $J_B$ and $J_A$ are the modular conjugations of the corresponding algebras.

Euclidean path-integrals offer an intuitive and geometric way of seeing why Theorem \ref{thmRinlder} holds. The operator $\Delta_{W(a)}^\alpha$ for real $\alpha$ is rotation around $x^1=a$ by angle $2\pi \alpha$.
It is simply a geometric observation that the successive application of two rotations by $\pi/2$ in opposite directions around the two entangling surfaces, i.e. $T=\Delta_{W(0)}^{1/4}\Delta_{W(a)}^{-1/4}$ results in a translation in Euclidean time by $a$; see Figure \ref{fig1}. Repeating the argument with a general angle $\alpha$, we obtain the expression in (\ref{T_a2}). 

The length scale $a$ parameterizes the relative size of the inclusion. The small $a$ limit of $T(z)T(-z)^{-1}$ for imaginary $z$ corresponds to the high-temperature Gibbs state. As we send $a\to 0$, we obtain the identity operator, whose one-norm is infinite. Therefore, a small $a$ can be viewed as a UV regulator of the BW form. However, note that in the case of flat space, the finite temperature partition function is infrared divergent due to the infinite volume of spacetime, and the extensivity of free energy. Another important observation is that only the combination $a\sin\alpha$ appears in the characteristic function. Therefore, an alternate high-temperature limit is to keep $a$ fixed and send $z=i\alpha\to0$. We will come back to this in Section \ref{sec:CFT}.

 \begin{figure}[t]
     \centering
     \includegraphics[width=0.8\linewidth]{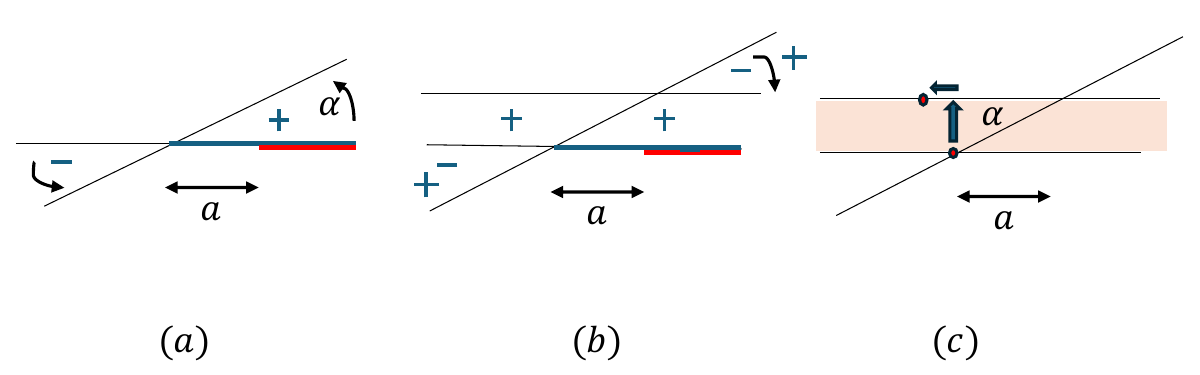}
     \caption{\small{(a) Euclidean modular flow $\Delta_W^\alpha$ is a rotation by angle $\alpha$ on the Euclidean plane around $x=0$. (b) $\Delta_{W(a)}^{-\alpha}$ rotates back by the same angle, but around $x=a$. (c) The new result is a Euclidean strip of width $a\sin(2\pi\alpha)$. (c) The end point of $W(0)$ (red-dot) has moved in the $x^1$ direction by $a\sin(2\pi\alpha)$ in the Euclidean time direction combined with the translation by $a(1-\cos(2\pi \alpha))$.}}
     \label{fig1}
 \end{figure}

\section{CFT Hamiltonian from a ball}\label{sec:CFT}

In vacuum CFT, once again, we can write down the Minkowski Hamiltonian in terms of the local data of a ball-shaped region. The intuitive way to see it is this:
\begin{eqnarray}
    K_{B(R)}=2\pi\int dr d\Omega_{d-2}\: \frac{(R^2-r^2)}{2R}T_{00}(x)
\end{eqnarray}
Therefore,
\begin{eqnarray}\label{HMinkCFT}
    R\p_R K_{B(R)}=2\pi \int \frac{R^2+r^2}{2R}T_{00}(x)
\end{eqnarray}
Adding the two terms gives the Hamiltonian in Minkowski space:
\begin{eqnarray}
    (R\p_R +1)K_{B(R)}=2\pi R \int T_{00}(x)
\end{eqnarray}
Note that this operator has a continuous spectrum. 

In CFTs, we extend the correlators to the Lorentzian cylinder via the embedding map $\tan(T\pm \theta/2)=x^0\pm r$ which maps $x^0=0$ surface to $T=0$ with $\tan(\theta/2)=r$ of the cylinder  \cite{luescher1975global}; see figure \ref{fig2} . Consider the vacuum state of a CFT in $\mathbb{R}^{d-1,1}$ conformally embedded inside the Lorentzian cylinder $\mathbb{S}^{d-1}\times \mathbb{R}_T$ with $T$ the global time on the cylinder\footnote{The universal cover of the conformal group $\widetilde{SO(d+1,2)}$ acts on the Lorentzian cylinder.}. The global Hamiltonian of the cylinder whose discrete spectrum contains the data of the conformal primaries is \footnote{See Appendix \ref{app:CFT} for a detailed discussion of our notation.}
\begin{eqnarray}\label{globalHam}
J_{(-1)0}&&=\int_{0}^\pi d\theta d\Omega_{d-2}\: T_{00}(\theta)\nn\\
&&=\frac{1}{2}\int dr d\Omega_{d-2}\:(1+r^2) T_{00}(x)
\end{eqnarray}
To isolate this Hamiltonian, we need to compute 
\begin{eqnarray}
    J_{(-1)0}=\lb \frac{R^2-1}{2R}-\frac{R^2+1}{2}\p_R\rb K_{B(R)}\ .
\end{eqnarray}
Theorem \ref{theoremCFT} makes the argument above rigorous.

 \begin{figure}[t]
     \centering
     \includegraphics[width=0.8\linewidth]{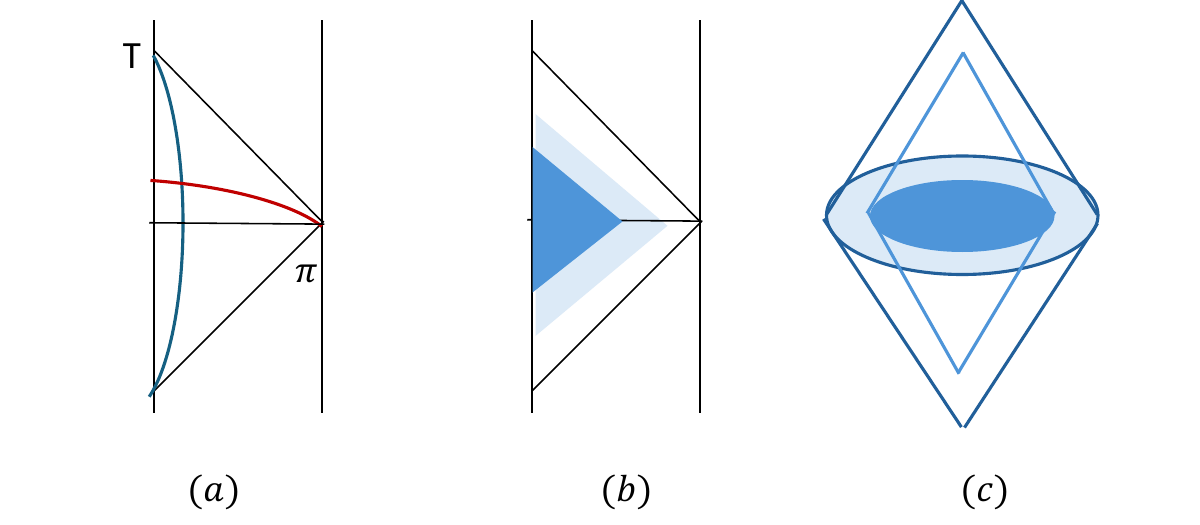}
     \caption{\small{(a) Extension of CFT from Minkowski space to Lorentzian cylinder. The blue and red lines are surfaces of constant $x^0$ and $r$ in Minkowski coordinates, and the range of $\theta\in(0,\pi)$. (b) The inclusion of concentric balls $B(e^{-2\pi (s+ u)})\subset B(e^{-2\pi u})$ with $s>0$ on Lorentzian cylinder. (c) The inclusion of concentric balls in Minkowski space.}}
     \label{fig2}
 \end{figure}

Consider the inclusion of ball-shaped regions $B_{u+s}\subset B_{u}$ centered at the $x^\mu=0$ where the radius of $B_{u+s}$ is $e^{-2\pi (s+u)}$ and $u\in \mathbb{R}$ and $s>0$. 
\begin{theorem}\label{theoremCFT}
    The characteristic function of the inclusion $B_{u+s}\subset B_u$ for $\Im (z)\in (-1/2,0)$ is
    \begin{eqnarray}
        &&\log T_{u,s}(z)=-(\beta_BK_B+\beta_H K_H+\beta_F K_F)=-\beta \lb\frac{1}{2\pi}\p_uK_{B_{u+s/2}} +\tanh(\pi z)K_F\rb\nn\\
        &&\beta_F=\beta\cosh(\pi s)\tanh(\pi z),\qquad\beta_B=\beta \sinh(\pi (2u+s)),\qquad 
        \beta_H=\beta\cosh(\pi (2u+s)),\nn\\
        &&\beta=-\text{sign}(z)\frac{i\cos^{-1}(A)  \cosh{(\pi z)}}{\pi\sqrt{1-\sinh^2(\pi z)\sinh^2(\pi s)}}= -\frac{\cos^{-1}(A)}{\sqrt{1-A^2}} \frac{2i}{\pi}\sinh(\pi s)\sinh(\pi z)\cosh(\pi z) \nn\\
        &&A=1-2\sinh^2(\pi s)\sinh^2(\pi z)
    \end{eqnarray}
where $B=B_0$ is a unit ball, $e^{-2\pi s}=|B_{u+s}|/|B_u|$ is the relative size of the inclusion and $|B|$ is the linear size of $B$, $K_H = 2\pi H$, and we have used 
\begin{eqnarray}
        &&K_{B_u}=\cosh(2\pi u)K_B+\sinh(2\pi u) K_H
\end{eqnarray}
and $H=J_{0(-1)}$ is the global Hamiltonian of the cylinder. The boundary value at $z=-i/2$ is
    \begin{eqnarray}
        \log T_{u,s}(-i/2)= -2 s i K_F\ .
    \end{eqnarray}
\end{theorem}
See Appendix \ref{app:CFT} for a proof. Note that for real $z$ under the transformation $z\to -z$, $\beta$ gets a minus sign, but the coefficient of $K_F$ is invariant. We have the following corollary
\begin{corollary}
For general positive values of $s$ and real $z$ with $\sinh^2(\pi s)\sinh^2(2\pi z) < 1$, we have
\begin{align}
    &\log T_{u,s}(z)T_{u,s}(-z)^{-1} \nn\\
    &= i\kappa(\sinh(\pi s) \cosh(2\pi z) K_{B_u} + \cosh(\pi s) \tilde{H}_u),\\
    &\kappa = \frac{\log(C+i\sqrt{1-C^2})}{2\pi i\sqrt{1-C^2}} 4\sinh(\pi s) \sinh(2\pi z),\\
    &C = 1 -2 \sinh^2(\pi s)\sinh^2(2\pi z)\ .
\end{align}    
\end{corollary}
We are particularly interested in the limit $s\ll 1$ and $s e^{\pi \Re z}\ll 1$. In this limit, we define 
\begin{eqnarray}
    \dot{\beta}=(\p_s \beta)_{s=0}=-i\sinh(2\pi z)
\end{eqnarray}
and $T_{u,s}(z)=1+i M_{u,s}(z)$ so that 
\begin{eqnarray}
    &&M_{u,s}(z)=2 s \sinh(\pi z)\lb \cosh(\pi z)\frac{\p_u K_{B_u}}{2\pi}+\sinh(\pi z)K_F\rb+O(s^2),\nn\\
    &&\frac{1}{2\pi}\p_u K_{B_u}=\sinh(2\pi u)K_B+\cosh(2\pi u) K_H\ .
\end{eqnarray}
In this limit, we find
\begin{eqnarray}
    T(z)T(-z)^{-1}&&\simeq 1+i(M_{u,s}(z)-M_{u,s}(-z))\nn\\
    &&=1 + i 2s \sinh(2\pi z)\frac{\p_u K_{B_u}}{2\pi}=1-\frac{s \dot{\beta}}{\pi}\p_u K_{B_u}\nn\\
    \tilde{H}_u(z)&&=\frac{1}{2\dot{\beta}}\p_s\log\lb T_{u,s}(z)T_{u,s}(-z)^{-1}\rb_{s=0}\nn\\
    &&=\frac{{-}1}{2\pi}\p_u K_{B_u}=R \p_R K_{B(R)}={-}(\sinh(2\pi u)K_B+\cosh(2\pi u)K_H)\nn\ .
\end{eqnarray}
Note that, importantly for us, the operator $\tilde{H}_z$ is independent of $z$. 
Furthermore, the combination $s \sin(2\pi i z)$ appears in the expansion of $T(z)T(-z)^{-1}$. This means that we can compensate for a change in $s$ by a change in $z$. This is analogous to what we encountered in the inclusion of wedges in Section \ref{sec:Rindler}. 

We can construct both the global Hamiltonian $J_{(-1)0}$ and the generator of rotations $J_{0(d+1)}$ on the Lorentzian cylinder in terms of the local modular Hamiltonian
       \begin{eqnarray}
    2\pi \begin{pmatrix}
    J_{(-1)0}\\
        J_{0(d+1)}
    \end{pmatrix}
    &&={(-1)}\begin{pmatrix}
        &\cosh(2\pi u) & \sinh(2\pi u)  \\
         &\sinh(2\pi u)& \cosh(2\pi u)
    \end{pmatrix}
    \begin{pmatrix}
        & \tilde{H}_u \\
        & -K_{u}
    \end{pmatrix}\nn\\
    &&=\frac{{+}1}{2R}\begin{pmatrix}
        &1+R^2 & 1-R^2  \\
         &1-R^2& 1+R^2
    \end{pmatrix}
    \begin{pmatrix}
        & R\p_R K_{B(R)}\\
        & K_{B(R)}
    \end{pmatrix}\ .
\end{eqnarray}
Defining charges that correspond to simultaneous time evolution and rotation in $\theta$ we find
\begin{eqnarray}
    2\pi(J_{(-1)0}\pm J_{0(d+1)})={-}e^{\pm 2\pi u}\lb\tilde{H}_u\mp K_u \rb={+}R^{\mp 1}\lb R\p_R\pm 1\rb K_{B(R)}\ .
\end{eqnarray}
This is a result of the pair of modular future and past algebras in vacuum CFT  (positive and negative half-sided modular inclusions) \cite{ouseph2024local}.
In the limit of very large and very small regions, we find
\begin{eqnarray}
    &&2\pi H_{R\ll 1}\sim \frac{{+}1}{2R}(R\p_R+1)K_{B(R)}\nn\\
    &&2\pi H_{1\ll R}\sim \frac{{-}R}{2}(R\p_R-1)K_{B(R)}\label{Minkwoski}\ .
\end{eqnarray}
The first expression matches (\ref{HMinkCFT}) up to an extra factor of $1/2$ which comes from (\ref{globalHam}).


Our result of Theorem \ref{theoremCFT} solely relies on the relations between the modular Hamiltonian and conformal generators on the Lorentzian cylinder. This relation remains the same on a Lorentzian cylinder that is conformally flat at every point \cite{frob2023modular}. Therefore, our theorem extends our theorem to conformally flat spacetimes. 

\section{Local approximation to global QFT Hamiltonian}\label{sec:QFT}

Here, we use our CFT results of the previous section to propose local approximations to the global Minkowski Hamiltonian in QFT or excited states of the CFT. In particular, our proposal is the following generalization of \ref{Minkwoski} as local approximations of the global Minkowski Hamiltonian 
\begin{eqnarray}
    2\pi H_z(R)=\frac{1}{R}\lb \mathcal{E}_z(R\p_R K_{B(R)})+K_B(R)\rb
\end{eqnarray}
where we further apply a coarse-graining map $\mathcal{E}_z$ to $R\p_R K_{B(R)}$. In the case of vacuum CFT, this coarse-graining becomes the trivial identity map, and we recover (\ref{Minkwoski}). Our proposal follows from a direct calculation of (\ref{proposal}) in QFT:
\begin{eqnarray}
    \tilde{H}_R(z)&&=\frac{\pi}{\sin(2\pi i z)}\Delta_{B(R)}^{i z}R\p_R\lb \Delta_{B(R)}^{-2iz} \rb \Delta_{B(R)}^{iz}=\mathcal{E}_z(R\p_R K_{B(R)})\ .
\end{eqnarray}
\begin{theorem}\label{theoremQFT}
Consider the net of inclusion of ball-shaped regions $B(e^{-2\pi s}R)\subset B(R)$ with $s>0$ and $K(s)$ their corresponding modular Hamiltonians on a Hilbert space and the charateristic function $T_s(z)=\Delta^{iz}_{B(R)}\Delta^{-iz}_{B(e^{-2\pi s}R)}$ with $\Im(z)\in (-1/2,0)$. Then,
    \begin{eqnarray}
        \tilde{H}(z) &&\equiv \frac{{-1}}{2\sin(2\pi iz)}\p_s\log T_s(z)T_s(-z)^{-1}\Big|_{s=0}=\mathcal{E}_z(R\p_R K_{B(R)})\nn\\
        \mathcal{E}_z(R\p_R K_{B(R)})&&=\frac{\pi}{2\sinh(2\pi z)}\int_{-z}^{z} dv \, e^{i v \ad_{K(0)}} R\p_R K_{B(R)}\nn\\
        &&=\frac{2\pi \sin(z \ad_{K(0)})}{\sinh(2\pi z)\ad_{K(0)}} R\p_R K_{B(R)}\nn\\
        &&=\frac{{-}\pi}{\sinh(2\pi z)}\{\Delta_{B(R)}^{iz},R\p_R\Delta_{B(R)}^{-iz}\}\ .
    \end{eqnarray}
    where the adjoint action is defined as $\ad_A B = [A, B]$. In the limit of $z\to 0$ we have $\mathcal{E}_0(X)=\frac{1}{2\pi}X$.
\end{theorem}
The proof follows from Lemmas \ref{Htilde} and \ref{lemmaanitocmmuteApp} in appendix \ref{app:Ht}. The super-operator $\mathcal{E}_z(X)$, up to an overall $z$-dependent factor averages modular flown $X$, i.e. $X_z=e^{i z \ad_K}X$, over the range $(-z,z)$:
\begin{eqnarray}
&&\tilde{H}_R(z)=\mathcal{E}_z(R \p_R K_R)\nn\\
   && \mathcal{E}_z(X)=\frac{\pi}{\sinh(2\pi z)}\int_{-z}^z d\tau X_\tau\ .
\end{eqnarray}
In the limit $z\to 0$, we recover the CFT answer 
\begin{eqnarray}
    \tilde{H}_R(0)=R \p_R K_{B(R)}\ .
\end{eqnarray}
Another way to understand the operator $\tilde{H}_R(z)$ is that it captures the response of the double cone $\Delta^{iz}$ operator to a change in scale:
\begin{eqnarray}
    \tilde{H}_R(z)=\frac{{-}\pi}{\sinh(2\pi iz)}\{\Delta_R^{iz},R\p_R \Delta_R^{-iz}\}\ .
\end{eqnarray}

For a QFT in Minkowski space with a mass gap $m$, and the modular Hamiltonians of ball-shaped regions of radius $R$ with $e^{-2\pi s}R$, we propose the following set of operators as local approximations to the global dimensionless Hamiltonian of Minkowski space
\begin{eqnarray}
    H_x(z)=\mathcal{E}_z(\p_x K_{B(R)})+ x^{-1} K_{B(R)}
\end{eqnarray}
where $x=m R$. In a general QFT, $H_x(z)$ depends on $R$. For small regions $x\ll 1$, the reduced state on a small enough ball-shaped region tends to the vacuum CFT answer, and we obtain the Minkowski Hamiltonian of the UV fixed point. Below, we set up a perturbation theory around this fixed-point answer.

\subsection*{Conformal perturbation theory}

Consider a one-parameter family of vectors $\ket{\Omega(\lambda)}$ continuously connected to the vacuum of the CFT. This can be viewed either as excited states of CFT or the vacuum of conformal perturbation theory. 

\begin{lemma}\label{lemma:perturb}
Consider the net of inclusion of ball-shaped regions $B(e^{-2\pi s}R)\subset B(R)$ with $s>0$ and a one-parameter family of states $\ket{\Omega(\lambda)}$ such that the modular Hamiltonian $K(s, \lambda)$ is a two-parameter family of modular Hamiltonians on a Hilbert space, smooth in both $s$ and $\lambda$. The operator $\tilde{H}(z,\lambda)$ defined in theorem \ref{theoremQFT} admits the following expansion to first order in $\lambda$:
\begin{equation}
    \tilde{H}(z, \lambda) = \tilde{H}(z,0) + \lambda \tilde{H}'(z,0) + \mathcal{O}(\lambda^2)
\end{equation}
where 
\begin{align}
    \tilde{H}(z,0) &= \frac{{-}\sin(z \ad_K)}{\sinh(2\pi z) \ad_K} \dot{K}(0,0)\\
    \tilde{H}'(z,0) &= \frac{{-}\sin(z \ad_K)}{\sinh(2\pi z) \ad_K} (\dot{K}') {-} \frac{iz^2}{2\sinh(2\pi z)} \int_{-1}^{1} dy \, y \int_0^1 dw \, e^{izy(1-w)\ad_K} \Big[ K', \, e^{izyw\ad_K} \dot{K} \Big]\label{Hprimez0}
\end{align}
with $K = K(0,0)$, $\dot{K} = \p_s K=-2\pi R\p_R K$, $K' = \p_\lambda K$, $\dot{K}' = \p_\lambda \p_s K$. 
\end{lemma}
See Lemma \ref{lemmaperturbApp} in Appendix \ref{app:Ht} for a proof. We apply the lemma to the following two setups:
\begin{enumerate}    
 \item {\bf Excited states of CFT:} 

Instead of considering the vacuum of QFT, we can consider the excited states of CFT. More generally, for an inclusion $B(e^{-2\pi s}R)\subset B(R)$ and any pair of vectors $\Phi$ and $\Psi$ in the Hilbert space that are common cyclic and separating with respect to the inclusion (see Appendix \ref{subsection:common}), we can consider the following relative characteristic function (c.f. Appendix \ref{subsubsection:relativecharacteristicfunction})
\begin{eqnarray}
   T_{\Phi|\Psi}(z)=\Delta_{\Phi|\Psi;B(R)}^{iz}\Delta_{\Phi|\Psi;B(e^{-2\pi s}R)}^{-iz}
\end{eqnarray}
and motivated by Theorem \ref{theoremQFT}, define
\begin{eqnarray}
    \tilde{H}_{R;\Phi|\Psi}(z)=-\frac{1}{2\sinh(2\pi z)}\int_{-z}^{z}du\:e^{i u K_{\Phi|\Psi}}(\p_s K_{\Phi|\Psi})e^{-i u K_{\Phi|\Psi}}\ .
\end{eqnarray}
The simplest case to consider is the case of excited states $UU'\ket{\Omega}$ and $VV'\ket{\Psi}$ where $U$ and $V$ are unitaries supported inside $B(e^{-2\pi s}R)$ and $U'$ and $V'$ are supported outside $B(R)$. Then, we know that for both $B(R)$ and $B(e^{-2\pi s}R)$ we have \cite{lashkari2021modular}
\begin{eqnarray}
    &&\Delta_{\Phi|\Psi}=V'U\Delta_{\Omega;A/B}U^\dagger (V')^\dagger
\end{eqnarray}
which implies
\begin{eqnarray}
     &&T_{\Phi|\Psi;s}(-i\alpha)=\Delta_{\Phi|\Psi;B(R)}^{\alpha}\Delta_{\Phi|\Psi;B(e^{-2\pi s}R)}^{-\alpha}=V'U T_{\Omega|\Omega;s} U^\dagger (V')^\dagger\nn\\
     &&\tilde{H}_{R;\Phi|\Psi}(z)=V'U \tilde{H}_{R;\Omega|\Omega}U^\dagger (V')^\dagger\ .
\end{eqnarray}
To make the example more non-trivial, we consider a one-parameter family of states $\ket{\Phi(\lambda)}$ and $\ket{\Psi(\lambda)}$ created by the action of operators inside the small ball
\begin{eqnarray}
    &&\Phi_\lambda=1+\lambda \phi^{(1)}+\frac{\lambda^2}{2}\phi^{(2)}+\cdots\nn\\
    &&\Psi_\lambda=1+\mu\psi^{(1)}+\frac{\mu^2}{2}\psi^{(2)}+\cdots
\end{eqnarray}
on the vacuum. From \cite{lashkari2021modular,lashkari2023perturbation} we know that 
\begin{eqnarray}
    &&K_{\Phi|\Psi}-K_\Omega=\lambda K^{(10)}+\mu K^{(01)}+\lambda\mu K^{(11)}+O(\lambda^2)+O(\mu^2)\nn\\
    &&K^{(10)}=-\frac{1}{2}\mathcal{F}(\phi^{(1)}_{-i/2}+\phi^{(1)}_{i/2}),\nn\\
    &&K^{(01)}=\mathcal{F}(J\psi^{(1)}J),\qquad K^{(11)}=0\nn\\
    &&\mathcal{F}(X)=\pi\int_{-\infty}^\infty \frac{dt}{\cosh^2(\pi t)}e^{i t \ad_K}(X)=\frac{\pi\ad_K}{\sinh(\ad_K/2)}(X)\ .
\end{eqnarray}
As a result, for the one-parameter family of excited states $\ket{\Phi(\lambda)}$, we set $\phi^{(1)}=\psi^{(1)}$ and
$K'=K^{(10)}+K^{(01)}$. Plugging this into (\ref{Hprimez0}) of Lemma \ref{lemma:perturb} and taking $R\p_R$ derivatives gives the first-order correction to our Hamiltonian away from the fixed point.

\item  {\bf Conformal perturbation theory:}

Consider a CFT with Hamiltonian $H_0$ in $\mathbb{R}^{1,d-1}$, the vacuum state $\ket{\Omega_0}$, and the modular Hamiltonian of a ball-shaped region $B$ of radius $R$ which generates a geometric flow.
We deform the Hamiltonian by a relevant operator $\mathcal{O}$:
\begin{eqnarray}
    &&H=H_0+\lambda V\nn\\
    &&V=\int d^{d-1}x \:\mathcal{O}(\tau,x),\ .
\end{eqnarray}
The vacuum of the deformed theory, in the interaction picture, can be written using the Dyson formula as
\begin{eqnarray}
 \ket{\Omega(\lambda)}&&\sim \mathcal{T}\lb e^{-\lambda\int_{-\infty}^0  d\tau V(\tau)}\rb\ket{\Omega_0}\nn\\
 &&=\sum_{n\geq 0}\frac{(-1)^n}{n!}\int_{-\infty}^0d\tau_1\cdots d\tau_n \:\mathcal{T}\lb V(\tau_1)\cdots V(\tau_n)\rb\ket{\Omega_0} \nn\\
 &&V(\tau)=e^{H_0\tau}V e^{-H_0\tau}
\end{eqnarray}
where $\ket{\Omega_0}$ is the CFT vaccum. We use the expansion \cite{lashkari2023perturbation}
\begin{eqnarray}\label{logexpansion}
    \p_\lambda K&&=\frac{-\pi}{2}\int_{-\infty}^\infty \frac{dt}{\cosh^2(\pi t)}\Delta^{it}\lb \Delta^{-1/2}(\p_\lambda\Delta)\Delta^{-1/2}\rb \Delta^{-it}\nn\\
   && =\frac{-\pi}{2}\int_{-\infty}^\infty \frac{dt}{\cosh^2(\pi t)}\Delta^{it} \{\Delta^{-1},(\p_\lambda\Delta)\}\Delta^{-it}
\end{eqnarray}

In terms of the UV-regulated density matrix, we find
\begin{eqnarray}
    \rho^{-1}\delta_\lambda \rho=-\int d\tau\:\mathcal{T}\lb V(\tau)\rb\ .
\end{eqnarray}
Plugging this into the expansion of the logarithm in (\ref{logexpansion}) gives the first-order correction to the half-sided modular Hamiltonian of a ball-shaped region in conformal perturbation theory, see \cite{faulkner2015bulk} for explicit expressions. Next, we plug this $\p_\lambda K$ into Lemma \ref{lemma:perturb}. The final expression depends on the value of $z$, and does not simplify beyond the form in Lemma \ref{lemma:perturb}.
Schematically, it results in approximations of the global Hamiltonian that take the form
\begin{eqnarray}
    &&H_z=H_{CFT}+H'_z\lambda+O(\lambda^2)\nn\\
    &&H'_z=\int f_z(x,\tau)\mO(x,\tau),
\end{eqnarray}
where $H_{CFT}$ is the Minkowski CFT Hamiltonian which is $z$-independent, and the first-order correction takes the form of an integral with a single insertion of the relevant conformal primary operator $\mO$ in spacetime with some function $f_z(x,t)$. This general form is a consequence of continuity in $\lambda$. Since $\mO$ is a relevant deformation, intuitively, one expects that its high-energy physics of all $H_z$ are the same and match that of the actual Hamiltonian of the theory $H_{CFT}$. However, for no value of $z$, we reproduce the exact first-order expression for $H(\lambda)$. It is an interesting question to ask whether one can use the optimization of function $p(z)$ of the form
\begin{eqnarray}
    \int dz p(z) \mathcal{E}_z(X)
\end{eqnarray}
to make sure the local approximation matches the conformal perturbation theory at the first order in $\lambda$. We postpone an exploration of this and its potential connections to the universal recovery map to upcoming work.


We can compare the spectrum of $H_z$ operator with that of the CFT using non-degenerate perturbation theory:
\begin{eqnarray}
    (H_{CFT}-E_0+\lambda (H'_z-E_1))\lb \ket{n^{(0)}}+\lambda\ket{n^{(1)}}\rb=O(\lambda^2)\ .
\end{eqnarray}
The first-order change in eigenvalues and eigenvectors are given by 
\begin{eqnarray}
    E_1(n)=\int f_z(x,\tau)\bra{n^{(0)}}\mO(x,\tau)\ket{n^{(0)}}\nn\\
    \ket{n^{(1)}}=\int f_z(x,\tau)\sum_{k\neq n}\frac{\bra{k^{(0)}}\mO(x,\tau)\ket{n^{(0)}}}{(E_0(n)-E_0(k))}\ket{k^{(0)}}\ .
\end{eqnarray}
We are interested in the high-energy part of the spectrum of $H_z$. In a chaotic CFT, it follows from the eigenstate thermalization hypothesis that the matrix elements of light operator $\mO_p$ in these heavy energy eigenbasis, at high energies and for the vast majority of energy eigenstates, are given by 
\begin{eqnarray}
&&\bra{E_h}\mO_p\ket{E_{h'}}=\mO_p(E_h)\delta_{hh'}+O^p_{hh'}
\end{eqnarray}
where the first term is diagonal and depends only on the energy $E_h$, and the dependence of the microstate $h$ or $h'$ is all captured in the second term $O^p_{hh'}\sim e^{-S(E)}$ with $E=(E_h+E_{h'})/2$. This implies that the eigenvalues and eigenvectors of our local approximation $H_\lambda$, to the first order $\lambda$, do not behave erratically from one microstate to the next. 

\end{enumerate}

\section{Discussion}

The operator algebraic approach to QFT focuses on local data that is a net of von Neumann algebras corresponding to causally complete regions of spacetime. From the perspective of algebraic QFT, we start either with the algebra of a wedge $W$ or the causal development of a ball $B$, and a state that summarizes the entire knowledge of a local observer \footnote{Note that by the time-like tube theorem, the von Neumann algebras associated to an observer is the causal development of a ball.}. Then, the GNS representation constructs the purification of this state in some fictitious commutant $(\mA_W)'$ or $(\mA_B)'$ that is inaccessible to the observer. The Bisogniano-Wichmann theorem is the surprising statement in vacuum QFT, for the algebra of a wedge, this fictitious purification describes the physics of the causal complement of $W$. The local Rindler observer can use the inclusions inside the wedge (a pair of future and past modular subalgebras) to reconstruct translations in space and time, and describe operators anywhere in spacetime \cite{revolutionizing}. This intuition generalizes to ball-shaped regions of vacuum CFT on a Lorentzian cylinder. In this work, we described how a local observer who has access only to a finite-sized ball, in a CFT, can explicitly reconstruct the global Hamiltonian on the Lorentzian cylinder that encodes all the global data of the theory. We used this to propose local approximations to the global Hamiltonian in QFT. From the perspective of local physics and operator algebras, for an observer with access to a finite ball-shaped region in some state of QFT, the commutant created by the GNS construction is a fictitious universe that purifies the state of the observer. The global Hamiltonian we proposed are candidate global Hamiltonians in this fictitious purified state, to the best of observer's knowledge, in a Bayesian sense. Different values of $z$ are choices that maximize various (Renyi) entropic measures. These choices can be further refined using (\ref{averagemeasure}) to maximize other more physically relevant measures such as preserving various infrared data. 
For instance, one might require that the reconstruction is such that it matches the first-order correction in conformal perturbation theory, or one might require that in the vicinity of the entangling wedge, the proposed approximation becomes exact. We postpone a further exploration of this optimization to upcoming work.

We point out two potential applications of our results to QFT in cosmological spacetimes and quantum chaos. In QFT in curved spacetime background, in the absence of a global time-like Killing vector, there is no preferred choice of Hamiltonian, and hence no preferred choice of vacuum. In cosmological or black hole spacetimes, an observer might not have access to the entire spacetime manifold due to the presence of cosmological or event horizons. Our Hamiltonian $H(s)$ is constructed entirely from local data accessible to the observer provides a canonical choice of extending the spacetime beyond the horizon to define a sensible closed quantum system. The degrees of freedom behind the horizon $\mA'$ are the canonical purification of the local data $\mA$, and together they evolve under our global Hamiltonian. In particular, it would be interesting to study the implications of our discussion for QFT in the Poincare patch of de Sitter space, which is conformally flat, and our theorem \ref{theoremQFT} applies.

Another application of our result is that it allows for a local characterization of quantum chaos. For early-time measures of chaos, our proposed Hamiltonians can be used to regulate the modular Hamiltonian, and hence the modular spectral form factor. We postpone exploring the connection between the quasi-normal modes as discussed in \cite{chen2024comments}, and those of non-commutative Anosov systems, to future work. 

Most discussions of late-time quantum chaos rely on the spectral properties of the global Hamiltonian. For instance, at times of order $t\sim e^{S/2}$ the spectral form factor of quantum chaotic systems, appropriately coarse-grained, is expected to show a universal ramp whose origin is in the level repulsion in the spectrum of the global Hamiltonian. Whereas at times of order $e^S$, the spectral form factor is expected to plateau due to the discreteness in the spectrum of the global Hamiltonian. It is desirable to have a definition of quantum chaos that does not rely on the degrees of freedom far away, inaccessible to the observer. Our local reconstruction of the cylinder Hamiltonian in the CFT can be used to provide a local definition for late-time spectral measures of chaos that are sensitive to the discreteness of the spectrum at late-enough times. Note that it is a highly non-trivial result that in the case of CFT, the reconstructed global Hamiltonian has a discrete spectrum.

Our approximation of the global Hamiltonian using the local state of a topologically trivial region resonates well with the philosophy of entanglement bootstrap. Traditionally, the classification of low-energy phases of gapped systems is discussed using the space of gapped Hamiltonians. In the entanglement bootstrap, the key idea is that one classifies the phases using the reduced state in the vacuum, instead of the global Hamiltonian. Our work can be interpreted as a generalization of this principle to gapless systems.



Consider the eternal black hole in AdS in the limit of $G_N\to 0$. The boundary theory is a GFF with time-band algebras. Assuming spherical symmetry, we can view time-band algebras as time-interval algebras of $0+1$ GFF with an infinite tower of fields corresponding to the spherical harmonics. Then, the assumption of nuclearity implies that the split property, which in turn implies that the emergent type III$_1$ algebras of GFF are hyperfinite. This is related to the conjecture of Leutheusser and Liu \cite{leutheusser2023emergent}. It would be interesting to use our work to find constraints on the spectral density $\rho(\omega)$ of GFF required by the hyperfiniteness of the bulk algebra, refining previous work in \cite{furuya2023information}.


Instead of the inclusion of wedges, we can consider $W(\ep,x_\perp)\subset W(0)$ such that we have deformed the wedge inwards at a particular point $x_\perp$ on the entangling wedge. Then, instead of the Poincare algebra, we have to consider the algebra generated by ANEC operators and boost. One can do an analgous geometric deformation of ball that is rotation asymmetric. We postpone this to future work.

Finally, we make the following comments about our upcoming works: 1) Our operator $T(z)T(-z)^{-1}$ can be interpreted as a modular scattering matrix, and satisfies a modular chaos bound. 2) In lattice quantum systems with a continuum limit, one can compute our local approximations of the global Hamiltonian and compare with the actual lattice Hamiltonian. This approximation should become exact as we approach the fixed point.












\section*{Acknowledgements}
We thank Matthew Dodelson, Keiichiro Furuya, Victor Ivo, Albert Law, Hong Liu, Juan Maldacena, Mudassir Moosa, Shoy Ouseph, Laimei Nie, and Xiaoliang Qi for many insightful conversations. NL’s work is supported by the DOE grants DE-SC0007884 and DE-SC0025547, and the Templeton Foundation grant 63670 on ``Emergence of time from chaos in operator algebras".


\appendix

\section{Operator Algebraic Structures of $L^2$-Nuclearity}\label{appendix:operatoralgebraicstructure}
In this appendix, we introduce the necessary operator algebraic structures underlying the physics discussions in the main text. We start with a discussion on the existence and density of common cyclic and separating vectors. This provides the necessary backgrounds to introduce the characteristic function of von Neumann algebra inclusions. Finally, we introduce the $L^2$-nuclearity condition central to this paper.

From the outset, we comment on the basic assumptions that we will use throughout. These assumptions are standard and are satisfied by most well-behaved algebraic quantum field theories.
\begin{enumerate}
    \item All local algebras of observables are assumed to be of type III;
    \item In addition, all local algebras of observables are assumed to have separable preduals
\end{enumerate}
Notice that for our discussion, we do not need the axioms of algebraic quantum field theories.
\subsection{Common cyclic and separating vectors}\label{subsection:common}
In this subsection, we consider two type III von Neumann algebras $\mathcal{A},\mathcal{B}$ acting on a common separable Hilbert space $\mathcal{H}$. We would like to know:
\begin{enumerate}
    \item if there exists a common cyclic and separating vector;
    \item if there exists a dense subset of common cyclic and separating vectors.
\end{enumerate}
It turns out the answers to both questions are affirmative.

First recall the following fact\cite{blackadar}:
\begin{lemma}\label{lemma:blackadar}
    Every type III von Neumann algebra acting on a separable Hilbert space is in standard form.
\end{lemma}
Under our assumptions, both $\mathcal{A}$ and $\mathcal{B}$ acting on $\mathcal{H}$ are in standard forms. Hence, there exist two non-empty sets of cyclic and separating vectors $\mathcal{C}_\mathcal{A}$ and $\mathcal{C}_\mathcal{B}$ where $\mathcal{C}_\mathcal{A}$ is the set of cyclic and separating vectors for $\mathcal{A}$ and analogously $\mathcal{C}_\mathcal{B}$ for $\mathcal{B}$. The nontrivial questions are whether $\mathcal{C}_\mathcal{A}\cap\mathcal{C}_\mathcal{B} \neq\emptyset$ and whether $\mathcal{C}_\mathcal{A}\cap\mathcal{C}_\mathcal{B}$ is dense. 

The existence question is answered by the following observation due to Dixmier and Marechal\cite{dixmiermarechal}:
\begin{lemma}
    For a sequence of von Neumann algebras $(A_1,A_2,...)$ acting on the same Hilbert space $\mathcal{H}$, if each von Neumann algebra has a cyclic and separating vector, then there exists a common cyclic and separating vector for all $A_i$'s. 
\end{lemma}
This lemma directly shows that $\mathcal{C}_\mathcal{A}\cap\mathcal{C}_\mathcal{B} \neq \emptyset$. The density question explores whether such common cyclic and separating vectors are hard to find. The answer to this question follows from two simple observations:
\begin{lemma}
    For any von Neumann algebra $\mathcal{A}$ in its standard form acting on a Hilbert space $\mathcal{H}$, the set of cyclic vectors is dense in $\mathcal{H}$ and the set of separating vectors is dense in $\mathcal{H}$.
\end{lemma}
\begin{proof}
    Since $\mathcal{A}$ is in its standard form, there exists a cyclic vector $\ket{\xi}\in\mathcal{H}$. For any invertible operator $\mathcal{O}\in\mathcal{A}$, the vector $\mathcal{O}\ket{\xi}\in\mathcal{H}$ is also a cyclic vector:
    \begin{equation}
        \overline{A\mathcal{O}\ket{\xi}} = \overline{A\mathcal{O}^{-1}\mathcal{O}\ket{\xi}} = \overline{A\ket{\xi}} = \mathcal{H}
    \end{equation}
    
    Finally, by an observation due to Dixmier and Marechal \cite{dixmiermarechal}, any element in $\mathcal{A}$ is a strong-operator limit of a sequence of invertible elements in $\mathcal{A}$. Therefore the set $\{\mathcal{O}\ket{\xi}: \mathcal{O}\in \text{Inv}(\mathcal{A})\}$ is dense in $\mathcal{H}$ where $\text{Inv}(\mathcal{A})$ is the set of invertible operators in $\mathcal{A}$.

    Since $\mathcal{A}$ is in its standard form, its commutant $\mathcal{A}'$ acting on $\mathcal{H}$ is also in its standard form. Hence, the set of cyclic vectors of $\mathcal{A}'$ is dense in $\mathcal{H}$. Thus, the set of separating vectors of $\mathcal{A}$ is dense in $\mathcal{H}$.
\end{proof}
The next observation is again due to Dixmier and Marachel \cite{dixmiermarechal}. The statement applies to general von Neumann algebras.
\begin{lemma}
    For any von Neumann algebra $A$ acting on a Hilbert space $\mathcal{H}$ the set of cyclic vector forms a $G_\delta$ subset of $\mathcal{H}$ and the set of separating vector forms a $G_\delta$ subset of $\mathcal{H}$
\end{lemma}
The proof of this lemma is elementary and short. So we include it here. Recall a $G_\delta$ set of a topological space is the intersection of countably infinite open subsets. 
\begin{proof}
    An empty set is a $G_\delta$ set by definition.

    Suppose the set of cyclic vectors of $A$ is non-empty, i.e. there exists $\ket{\xi_0}\in \mathcal{H}$ such that $\overline{A\ket{\xi_0}}= \mathcal{H}$. Then for any operator $\mathcal{O}\in A$, we consider the countably infinite collection of open sets:
    \begin{equation}
        U(\mathcal{O},n):=\{\ket{\xi}\in\mathcal{H}:||\mathcal{O}\ket{\xi} - \ket{\xi_0}|| < \frac{1}{n}\}
    \end{equation}
    If $\ket{\xi}\in\mathcal{H}$ is a cyclic vector of $A$, then for all $n > 0$ there exists an operator $\mathcal{O}_n$ such that $||\mathcal{O}_n\ket{\xi} - \ket{\xi_0}|| < \frac{1}{n}$. Hence we have:
    \begin{equation}
        \ket{\xi}\in \cap_{n > 0}\cup_{\mathcal{O}\in A}U(\mathcal{O}, n)
    \end{equation}
    Note that because $\cup_{\mathcal{O}\in A} U(\mathcal{O}, n)$ is open, the intersection $\cap_{n > 0}\cup_{\mathcal{O}\in A}U(\mathcal{O}, n)$ is a $G_\delta$ set. 
    
    Conversely, any vector in this intersection is a cyclic vector of $A$. This is because the reference vector $\ket{\xi_0}$ is a cyclic vector. Thus for an arbitrary vector $\ket{\xi}\in\mathcal{H}$ and for any $\epsilon  > 0$, there exists an operator $\mathcal{O}\in A$ such that:
    \begin{equation}
        ||\mathcal{O}\ket{\xi_0} - \ket{\xi}|| < \frac{\epsilon}{2}
    \end{equation}
    For any vector $\ket{\xi'}$ in the aforementioned $G_\delta$ set, there exists another operator $\mathcal{O}'$ such that:
    \begin{equation}
        ||\mathcal{O}'\ket{\xi'
        } - \ket{\xi_0}|| < \frac{\epsilon}{2 ||\mathcal{O}||}
    \end{equation}
    Then we have:
    \begin{equation}
        ||\mathcal{O}\mathcal{O}'\ket{\xi'} - \ket{\xi}||\leq||\mathcal{O}||\cdot||\mathcal{O}\ket{\xi'} - \ket{\xi_0}|| + ||\mathcal{O}\ket{\xi_0} - \ket{\xi}|| <\epsilon
    \end{equation}
    Hence $\overline{\mathcal{A}\ket{\xi'}} = \mathcal{H}$. 

    Thus it follows that the set of cyclic vectors is exactly the $G_\delta$ set $\cap_{n >0}\cup_{\mathcal{O}\in A}U(\mathcal{O},n)$. Apply the same argument to the commutant $A'$ and since the set of separating vectors of $A$ is the set of cyclic vectors of $A'$, we arrive at the conclusion.
\end{proof}
Based on these two observations, the density of common cyclic and separating vectors follows from the Baire category theorem. 
\begin{corollary}
    For two type III von Neumann algebras $\mathcal{A}, \mathcal{B}$ acting on a common separable Hilbert space $\mathcal{H}$, then the set of common cyclic and separating vectors is a dense $G_\delta$ subset of $\mathcal{H}$.
\end{corollary}
\begin{proof}
    Following the previous two lemmas, for the algebra $\mathcal{A}$, the set of cyclic vectors is a dense $G_\delta$ subset of $\mathcal{H}$ and the set of separating vectors is also a dense $G_\delta$ subset of $\mathcal{H}$. By Baire category theorem, the intersection of these two sets is also a dense $G_\delta$ subset of $\mathcal{H}$, i.e. $\mathcal{C}_\mathcal{A}$ is a dense $G_\delta$ subset of $\mathcal{H}$. Similarly, $\mathcal{C}_\mathcal{B}$ is a dense $G_\delta$ subset of $\mathcal{H}$.

    Finally, by another application of Baire category theorem, the intersection $\mathcal{C}_\mathcal{A}\cap\mathcal{C}_\mathcal{B}$ is a dense $G_\delta$ subset of $\mathcal{H}$.
\end{proof}
\subsection{Characteristic functions of von Neumann algebra inclusions}\label{subsection:characteristicfunctions}
Under our assumptions, consider an inclusion of two type III von Neumann algebras $\mathcal{B}\subset \mathcal{A}$ where both algebras act on a separable Hilbert space $\mathcal{H}$. We have already seen that the set of common cyclic and separating vectors is dense in $\mathcal{H}$. Fix one such common cyclic and separating vector $\ket{\Omega}\in\mathcal{C}_\mathcal{A}\cap\mathcal{C}_\mathcal{B}$. Assume $\ket{\Omega}$ is normalized. Then $\bra{\Omega}\cdot\ket{\Omega}$ is a normal faithful state on both $\mathcal{A}$ and $\mathcal{B}$. Denote the state on $\mathcal{A}$ as $\omega_\mathcal{A}$ and the state on $\mathcal{B}$ as $\omega_\mathcal{B}$. By modular theory, we can define the corresponding modular operators and modular conjugations. Denote $\Delta_\mathcal{A}$ as the modular operator of $\omega_\mathcal{A}$ on $\mathcal{A}$ and $J_\mathcal{A}$ as the modular conjugation on $\mathcal{A}$. Denote $\Delta_\mathcal{B}$ as the modular operator of $\omega_\mathcal{B}$ on $\mathcal{B}$ and $J_\mathcal{B}$ as the modular conjugation on $\mathcal{B}$.

By definition of modular operator and modular conjugation, it is clear that $\Delta_\mathcal{A}$ is an extension of $\Delta_\mathcal{B}$:
\begin{equation}
    \mathcal{S}_\mathcal{A}\mathcal{O}_\mathcal{B}\ket{\Omega} = J_\mathcal{A}\Delta_\mathcal{A}^{1/2}\mathcal{O}_\mathcal{B}\ket{\Omega} = \mathcal{O}_\mathcal{B}^*\ket{\Omega} = J_\mathcal{B}\Delta_\mathcal{B}^{1/2}\mathcal{O}_\mathcal{B}\ket{\Omega} = \mathcal{S}_\mathcal{B}\mathcal{O}_\mathcal{B}\ket{\Omega}
\end{equation}
where $\mathcal{O}_\mathcal{B}$ can be any operator in $\mathcal{B}$. Hence, by general operator theory, we have operator inequality:
\begin{equation}
    \Delta_\mathcal{B}\geq \Delta_\mathcal{A}
\end{equation}
Since the function $f(z)  = z^\alpha$ for $0\leq \alpha\leq 1$ is an operator monotone function, then we have:
\begin{equation}
    \Delta_\mathcal{B}^\alpha \geq \Delta_\mathcal{A}^\alpha
\end{equation}
In particular, this means that for any $0\leq \alpha\leq 1$, the domain $\mathcal{D}(\Delta_\mathcal{B}^\alpha)$ is contained in the domain $\mathcal{D}(\Delta_\mathcal{A}^\alpha)$. Hence, the following operator is well-defined and closable for $0\leq \alpha\leq 1/2$:
\begin{equation}
    \Delta_\mathcal{A}^\alpha \Delta_\mathcal{B}^{-\alpha}
\end{equation}
In addition, using the fact that $f(z) = z^\alpha$ for $0\leq \alpha\leq 1$ is operator monotone and $g(z) = z^{-1}$ is operator anti-monotone for positive operators, we have:
\begin{equation}
    \Delta_\mathcal{A}^{\alpha}\Delta_\mathcal{B}^{-2\alpha}\Delta_{\mathcal{A}}^{\alpha} \leq \Delta_\mathcal{A}^\alpha\Delta_\mathcal{A}^{-2\alpha}\Delta_\mathcal{A}^\alpha = 1
\end{equation}
We are now ready to introduce the characteristic function of the inclusion $\mathcal{B}\subset\mathcal{A}$\cite{revolutionizing}:
\begin{definition}
    Given a pair of von Neumann algebras $\mathcal{B}\subset\mathcal{A}$ and a common cyclic and separating vector $\ket{\Omega}\in\mathcal{H}$, the characteristic function is defined as:
    \begin{equation}
        \mathcal{D}_{\mathcal{A}, \mathcal{B}}(t) := \Delta_\mathcal{A}^{it}\Delta_\mathcal{B}^{-it}
    \end{equation}
\end{definition}
We have already seen that the characteristic function $\mathcal{D}_{\mathcal{A}, \mathcal{B}}(t)$ can be extended to the imaginary values: $\mathcal{D}_{\mathcal{A}, \mathcal{B}}(-i\alpha)$ where $0\leq \alpha\leq 1/2$. The crucial fact is that the characteristic function can be extended to the entire strip $S(0, -1/2) := \{z\in\mathbb{C}: 0\geq \Im{z}\geq -1/2\}$ \cite{revolutionizing}:
\begin{lemma}\label{lemma:characteristic}
    Given a pair of von Neumann algebra $\mathcal{B}\subset\mathcal{A}$ with a common cyclic and separating vector $\ket{\Omega}$, then the characteristic function $\mathcal{D}_{\mathcal{A}, \mathcal{B}}(t)$ can be extended to the strip $\mathcal{S}(0, -1/2)$. The extension is a bounded analytic function on $\mathcal{S}(0,-1/2)$.

    Moreover, on the boundary $\partial\mathcal{S}(0,-1/2) = \{z\in\mathbb{C}:\Im{z} = 0\}\sqcup\{z\in\mathbb{C}:\Im{z} = -1/2\}$, the characteristic function is strongly continuous.
\end{lemma}
For a full discussion of the characteristic function, we refer the reader to \cite{revolutionizing}. Here we list the properties of the characteristic function \cite{revolutionizing}:
\begin{enumerate}
    \item $\mathcal{D}_{\mathcal{A},\mathcal{B}}(t)$ and $\mathcal{D}_{\mathcal{A}, \mathcal{B}}(t - \frac{i}{2})$ are unitary and strongly continuous in $t$;
    \item $\mathcal{D}_{\mathcal{A},\mathcal{B}}(z)$ is bounded analytic on the strip $\mathcal{S}(0,-\frac{1}{2})$;
    \item $\mathcal{D}_{\mathcal{A},\mathcal{B}}(t)\ket{\Omega} = \ket{\Omega}$;
    \item The cocycle condition holds: $\mathcal{D}_{\mathcal{A},\mathcal{B}}(t+s) = \sigma_t^{\ket{\Omega}}(\mathcal{D}_{\mathcal{A},\mathcal{B}}(s))\mathcal{D}_{\mathcal{A},\mathcal{B}}(t)$ where $\sigma_t^{\ket{\Omega}}$ is the modular automorphism on $\mathcal{A}$ associated with the state $\ket{\Omega}$;
    \item $\mathcal{D}_{\mathcal{A},\mathcal{B}}(t - \frac{i}{2}) = J_\mathcal{A}\mathcal{D}_{\mathcal{A},\mathcal{B}}(t)J_\mathcal{B}$. In particular, $\mathcal{D}_{\mathcal{A},\mathcal{B}}(-\frac{i}{2}) = J_\mathcal{A}J_\mathcal{B}$. This unitary operator is crucial in subfactor theory creating a chain of algebras:
    \begin{equation}
        \mathcal{A}\supset\mathcal{B}\supset\mathcal{D}_{\mathcal{A},\mathcal{B}}(-\frac{i}{2})^*\mathcal{A}\mathcal{D}_{\mathcal{A},\mathcal{B}}(-\frac{i}{2})\supset\mathcal{D}_{\mathcal{A},\mathcal{B}}(-\frac{i}{2})^*\mathcal{B}\mathcal{D}_{\mathcal{A},\mathcal{B}}(-\frac{i}{2})\supset...
    \end{equation}
    \item For all $t\in\mathbb{R}$, $\mathcal{D}_{\mathcal{A},\mathcal{B}}(t)\mathcal{D}_{\mathcal{A},\mathcal{B}}(-\frac{i}{2})^*\mathcal{A}\mathcal{D}_{\mathcal{A},\mathcal{B}}(-\frac{i}{2})\mathcal{D}_{\mathcal{A},\mathcal{B}}(t)^*\subset\mathcal{A}$ holds.
\end{enumerate}
Any operator-valued function associated with a pair of von Neumann algebras that satisfy the above-listed properties is called a characteristic function \cite{revolutionizing}.

For our discussion on $L^2$-nuclearity, we actually only need to consider $\mathcal{D}_{\mathcal{A},\mathcal{B}}(-i/4) = \Delta_\mathcal{A}^{1/4}\Delta_\mathcal{B}^{-1/4}$. However, the important observation regarding characteristic functions is the following observation \cite{revolutionizing}:
\begin{theorem}
    Given a von Neumann algebra $\mathcal{A}$ and a cyclic and separating vector $\ket{\Omega}$, then each subalgebra $\mathcal{B}\subset\mathcal{A}$ such that $\ket{\Omega}$ is cyclic for $\mathcal{B}$ is uniquely associated with a characteristic function.
\end{theorem}
Again for details, we refer the readers to \cite{revolutionizing}. The key point is that the characteristic function fully characterizes inclusions of von Neumann algebras with common cyclic vectors.

\subsubsection{Relative characteristic functions of von Neumann algebra Inclusions}\label{subsubsection:relativecharacteristicfunction}
Under our assumptions, consider an inclusion of two type III von Neumann algebras $\mathcal{B}\subset\mathcal{A}$ where both algebras act on a separable Hilbert space $\mathcal{H}$. Since the set of common cyclic and separating vectors is dense in $\mathcal{H}$, we can fix two common cyclic and separating states $\Phi$ and $\Psi$. Let $\ket{\Omega_\Phi}$ be the vector representing $\Phi$ and let $\ket{\Omega_\Psi}$ be the vector representing $\Psi$. Instead of modular operators, we consider relative modular operators $\Delta_{\Phi|\Psi,\mathcal{A}}$ and $\Delta_{\Phi|\Psi,\mathcal{B}}$. On the dense subspace $\mathcal{B}\ket{\Omega_\Psi}$, we have the simple identity:
\begin{equation}
    J_\mathcal{A}\Delta^{1/2}_{\Phi|\Psi,\mathcal{A}}\mathcal{O}_\mathcal{B}\ket{\Omega_\Psi} = \mathcal{O}_\mathcal{B}^*\ket{\Omega_\Phi} = J_\mathcal{B}\Delta^{1/2}_{\Phi|\Psi,\mathcal{B}}\mathcal{O}_\mathcal{B}\ket{\Omega_\Psi}
\end{equation}
where $\mathcal{O}_\mathcal{B}$ can be any operator in $\mathcal{B}$. Hence by the same argument used to show that the characteristic function is well-defined, the following operator is well-defined and closable for $0\leq \alpha\leq 1/2$:
\begin{equation}
    \Delta^\alpha_{\Phi|\Psi,\mathcal{A}}\Delta^{-\alpha}_{\Phi|\Psi,\mathcal{B}}
\end{equation}
Again these operators satisfy the following monotone relation:
\begin{equation}
    \Delta_{\Phi|\Psi,\mathcal{A}}^\alpha \Delta_{\Phi|\Psi,\mathcal{B}}^{-2\alpha}\Delta_{\Phi|\Psi,\mathcal{A}}^\alpha\leq 1
\end{equation}
Using the following identity:
\begin{equation}
    J_\mathcal{A}\Delta_{\Phi|\Psi,\mathcal{A}}J_\mathcal{A} = \Delta_{\Psi|\Phi,\mathcal{A}}^{-1}
\end{equation}
We can show the following identity:
\begin{align}
    \begin{split}
        J_\mathcal{A}\Delta_{\Phi|\Psi,\mathcal{A}}^{it+1/2}\Delta_{\Phi|\Psi,\mathcal{B}}^{-it-1/2}J_\mathcal{B} &= \Delta_{\Psi|\Phi,\mathcal{A}}^{it}J_\mathcal{A}\Delta_{\Phi|\Psi,\mathcal{A}}^{1/2}\Delta_{\Phi|\Psi,\mathcal{B}}^{-1/2}J_\mathcal{B}\Delta_{\Psi|\Phi,\mathcal{B}}^{-it} = \Delta_{\Psi|\Phi,\mathcal{A}}^{it}\Delta_{\Psi|\Phi,\mathcal{B}}^{-it}
    \end{split}
\end{align}

We can now define the relative characteristic function of the inclusion $\mathcal{B}\subset\mathcal{A}$:
\begin{definition}
    Given a pair of von Neumann algebras $\mathcal{B}\subset\mathcal{A}$ and two common cyclic and separating vectors $\ket{\Omega_\Phi}\in\mathcal{H}$ and $\ket{\Omega_\Psi}\in\mathcal{H}$, the relative characteristic function is defined as:
    \begin{equation}
        \mathcal{D}_{\Phi|\Psi,\mathcal{A},\mathcal{B}}(t) := \Delta_{\Phi|\Psi,\mathcal{A}}^{it}\Delta_{\Phi|\Psi,\mathcal{A},\mathcal{B}}^{-it}
    \end{equation}
\end{definition}
The discussions above already shown that the relative characteristic function is unitary for $t\in\mathbb{R}$ and for $t-i/2\in\mathbb{R}-i/2$. And the relative characteristic function can be extended to $\mathcal{D}_{\Phi|\Psi,\mathcal{A},\mathcal{B}}(z)$ where $\Im(z) \in \mathcal{S}(0,-1/2)$. Then the same proof of Lemma \ref{lemma:characteristic} can be used to show that the following:
\begin{lemma}
    Given a pair of von Neumann algebras $\mathcal{B}\subset\mathcal{A}$ with a pair of common cyclic and separating vectors $\ket{\Omega_\Phi}$ and $\ket{\Omega_\Psi}$, then the relative characteristic function $\mathcal{D}_{\Phi|\Psi,\mathcal{A},\mathcal{B}}(t)$ can be extended to the strip $\mathcal{S}(0,-1/2)$. The extension is a bounded analytic function on $\mathcal{S}(0,-1/2)$.
\end{lemma}
We do not need other properties of relative characteristic function in this work, so we postpone a more complete discussion to a later work.
\subsection{$L^2$-Nuclearity Condition}\label{subsection:l2nuclearity}
We are now ready to present the definition of $L^2$-nuclearity.
\begin{definition}
    Given a pair of von Neumann algebras $\mathcal{B}\subset\mathcal{A}$ with a common cyclic and separating vector $\ket{\Omega}\in\mathcal{H}$, $\mathcal{D}_{\mathcal{A}, \mathcal{B}}(-i/4) = \Delta_\mathcal{A}^{1/4}\Delta_\mathcal{B}^{-1/4}$ is a well-defined map between the standard Hilbert spaces:
    \begin{equation}
        \Delta_\mathcal{A}^{1/4}\Delta_\mathcal{B}^{-1/4}: L_2(\mathcal{B})\rightarrow L_2(\mathcal{A})
    \end{equation} 
    If $\Delta_\mathcal{A}^{1/4}\Delta_\mathcal{B}^{-1/4}$ is a nuclear operator, then we say the triplet $(\mathcal{B}\subset\mathcal{A}, \ket{\Omega})$ satisfies the \textit{$L^2$-nuclearity condition}.
\end{definition}
For a map between two separable Hilbert spaces, it is nuclear if and only if it is trace-class. For maps between general Banach spaces, the notion of a nuclear operator generalizes the notion of a trace-class operator \cite{nuclearandsumming}.

For algebraic quantum field theories, one of the most important consequences of $L^2$-nuclearity is the following observation \cite{buchholz2007nuclearity}:
\begin{theorem}
    Given a pair of von Neumann algebras $\mathcal{B}\subset\mathcal{A}$ with a common cyclic and separating vector $\ket{\Omega}$, if it satisfies $L^2$-nuclearity condition, then the inclusion is split. More specifically, we have isomorphism:
    \begin{equation}
        \mathcal{B}'\vee\mathcal{A}\cong\mathcal{B}'\otimes\mathcal{A}
    \end{equation}
\end{theorem}
For other important consequences of $L^2$-nuclearity, we refer the readers to \cite{buchholz2007nuclearity}. 

\subsubsection{$L^2$-Nuclearity and discrete spectrum in CFT}\label{subsubsection:discrete}
In the context of algebraic quantum field theory, $L^2$-nuclearity is usually proved using representation theory. For example, in the context of massless Lorentzian conformal field theories of dimension $d + 1$, the global conformal group $SO(d + 1, 2)$ admits (holomorphic \footnote{For our discussion on the discreteness of the spectrum of conformal Hamiltonian, holomorphicity is not essential. However, it is important when constructing modular conjugation.})  discrete series representations \cite{knapp}. By the well-known Harish-Chandra rank conditions \cite{knapp}, the existence of discrete series representations is equivalent to the condition that the global conformal symmetry group admits a compact Cartan subgroup. By definition, this compact Cartan subgroup is compact and Abelian (i.e. a torus) \footnote{It should be noted that such a (maximal) torus is not unique. In fact, any semi-simple Lie group admits many maximal tori, all of which are conjugate to one another. For our discussion, we fix a maximal torus.}. By the general representation theory of compact Abelian groups, any discrete series representation of the global conformal symmetry group can be decomposed into a direct sum of one-dimensional irreducible representations of the chosen compact Cartan subgroup. Hence the operators in the compact Cartan subgroup have discrete spectra and the direct sum decomposition is the eigen-space decomposition of these operators \footnote{The operators in the compact Cartan subgroup are unitary. The corresponding generators are unbounded self-adjoint operators, whose spectra are also discrete.}\footnote{The eigen-space decomposition exists because the generators of the unitaries in the compact Cartan subgroup have discrete spectra.}. Using root space decomposition, one can construct raising/lowering operators for each generator of the compact Cartan subgroup. Then it is not difficult to see that any irreducible discrete series representation must have lowest (or highest) weights. In turn, this shows that the generators of the compact Cartan subgroup must also have spectra bounded from below (or above). In physics, we typically assume the spectra are bounded from below.

As an example, the global conformal symmetry group $SL(2,\mathbb{R})$ of a conformal field theory on a null light ray (i.e. half of the degrees of freedom in a Lorentzian conformal field theory in $1+1$-dimension) admits a compact Cartan subgroup $SO(2)\subset SL(2, \mathbb{R})$. The generator of this $SO(2)$ is the conformal Hamiltonian, which admits discrete spectrum. As another higher-dimensional example, the global conformal symmetry group $SO(d+1 , 2)$ of a conformal field theory on $\mathbb{R}^{d,1}$ admits a compact Cartan subgroup isomorphic to $SO(d+1)\times SO(2)$. The generators of this Lie group forms a $1 + \frac{(d+1)d}{2}$ dimensional Lie algebra and all generators have discrete spectra. In both examples, the existence of compact Cartan subgroups ensures the existence of lowest (or highest) weight in a unitary irreducible representation.

\subsection{Nuclear norm}\label{subsection:defnuclear}
In this subsection, we briefly review the definition of nuclear norm. Nuclear norm of an operator between two Banach spaces generalizes the notion of trace norm of an operator between two Hilbert spaces \cite{nuclearandsumming}. 
\begin{definition}
    Let $T:X\rightarrow Y$ be an operator between two Banach spaces $X,Y$. The nuclear norm is defined as:
    \begin{equation}
        \nu(T) := \inf\{\sum_{i\in I}||f_i||\cdot||y_i||: T = \sum_{i\in I}y_if_i\text{ , }y_i\in Y\text{ , }f_i\in X^*\}
    \end{equation}
    where the infimum is taken over all such decompositions of $T$. 
\end{definition}
It is clear that not all operators between two Banach spaces have finite nuclear norm. Those operators with finite nuclear norms are called nuclear operators. Notice that in our definition, we do not require $X$ or $Y$ to be separable Banach space. This is important for application to algebraic quantum field theory because any infinite dimensional von Neumann algebra (e.g. algebra of local observables) is non-separable in norm topology. 

If both $X$ and $Y$ are finite dimensional, then nuclear norm is simply the sum of the absolute values of the singular values of $T$.

\section{Conformal field theory calculation}\label{app:CFT}
\subsection{Conformal symmetry on Lorentzian cylinder}
Consider a conformal field theory in $\mathbb{R}^{d,1}$ embedded in the Lorentzian cylinder $\mathbb{R}_t\times S^d$. The universal cover of the conformal group $\widetilde{SO(d+1,2)}$ acts on the Lorentzian cylinder. It is convenient to express the conformal Lie algebra $so(d+1,2)$ using the embedding coordinates $X^A=(X^{-1},X^0,X^i,X^{d+1})$ with the metric 
\begin{eqnarray}
    ds^2=g_{AB}dX^AdX^B=-(X^{-1})^2-(X^{0})^2+(X^i)^2+(X^{(d+1)})^2
\end{eqnarray}
\begin{eqnarray}
    &&[J_{AB},J_{CD}]=-i\lb g_{AC}J_{BD}+g_{BD}J_{AC}-g_{AD}J_{BC}-g_{BC}J_{AD}\rb\ .
\end{eqnarray}
As matrices we can represent the algebra generators as
\begin{eqnarray}
    (J_{AB})^C_{\: D}=i(\delta_A^{\:\: C}g_{BD}-\delta_B^{\:\: C}g_{AD})\ .
\end{eqnarray}
We have
\begin{eqnarray}
    J_{(-1)(d+1)}=D,\qquad J_{(-1)\mu}=\frac{-1}{2}(P_\mu+K_\mu),\qquad J_{\mu(d+1)}=\frac{-1}{2}\lb P_\mu-K_\mu\rb
\end{eqnarray}
where $P^\mu$ are translations, and $K^\mu$ are special conformal transformations. Note that in our notation $J_{(-1)0}=\frac{-1}{2}(P_0+K_0)$ generates global time translations in the Lorentzian cylinder.

The coordinates $x^\mu\in \mathbb{R}^{d,1}$ in the embedding coordinates take the form 
\begin{eqnarray}
    X(x)=\lb \frac{1+x^2}{2},x^\mu, \frac{1-x^2}{2}\rb
\end{eqnarray}
This makes sure that $X^+=X^{-1}+X^{d+1}=1$. To apply a conformal transfomation we apply a $SO(d,2)$ matrix, and project again to the surface at $X^+=1$:
\begin{eqnarray}
    (x')^\mu=\frac{(\Lambda\cdot X(x))^\mu}{(\Lambda\cdot X(x))^{-1}+(\Lambda\cdot X(x))^{(d+1)}}\ .
\end{eqnarray}
We can check explicitly that in our notation we have
\begin{eqnarray}
    &&\text{Boost}:\qquad \Lambda=e^{2\pi i s J_{01}}:\qquad 
\begin{cases}
(x')^0=\cosh(2\pi s)x^0-\sinh(2\pi s)x^1\\
(x')^1=\cosh(2\pi s)x^1-\sinh(2\pi s)x^0
\end{cases}\nn\\
    &&\text{Dilatation}:\qquad \Lambda=e^{2\pi i sJ_{-1(d+1)}}:\qquad (x')^\mu=e^{2\pi s} x^\mu\nn\\
    &&\text{Translation}:\qquad \Lambda=e^{ i a^\mu P_\mu}:\qquad (x')^\mu=x^\mu+a^\mu\nn\\
    &&\text{Special conf. transf.}:\qquad \Lambda=e^{i a^\mu K_\mu}:\qquad (x')^\mu=\frac{x^\mu+a^\mu x^2}{1+2a_\mu x^\mu+a^2 x^2}\ .
\end{eqnarray}

We are interested in two particular $so(2,1)\simeq psl(2,R)$ subalgebras. 
\paragraph{Case I $(X^{0},X^1,X^{d+1})$:} This subalgebra $psl(2,\mathbb{R})$ is generated by $\{J_{01},J_{0(d+1)},J_{1(d+1)}\}$ with commutators 
\begin{eqnarray}
    &&[J_{01},J_{1(d+1)}]=iJ_{0(d+1)}\nn\\
    &&[J_{01},J_{0(d+1)}]=iJ_{1(d+1)}\nn\\
    &&[J_{0(d+1)},J_{1(d+1)}]=-iJ_{01}\ . 
\end{eqnarray}
This is the same as 
\begin{eqnarray}
    [J_{01},(P_1-K_1)]=i(P_0-K_0),\qquad  [J_{01},(P_0-K_0)]=i(P_1-K_1),\qquad [P_0-K_0,P_1-K_1]=-4i J_{01}
\end{eqnarray}
where $J_{01}$ is the generator of boost in $X^1$-direction. Defining null momenta and special conformal transformations $P_\pm=\frac{1}{2}(P_0\pm P_1)$ and $K_\pm=\frac{1}{2}(K_0\pm K_1)$ we can write these equations as 
\begin{eqnarray}
    &&[J_{01},P_\pm -K_\pm]=\pm i(P_\pm-K_\pm)\nn\\
    &&[P_+-K_+,P_--K_-]=2i J_{01}\ .
\end{eqnarray}
It is convenient to define 
\begin{eqnarray}
&&L_0^{(I)}\equiv J_{01}\nn\\
&&L_\pm^{(I)}=-(J_{0(d+1)}\pm J_{1(d+1)})=(P_\pm-K_\pm)
\end{eqnarray}
 so that the commutation relations become 
\begin{eqnarray}
    [L^{(I)}_0,L^{(I)}_\pm]=\pm i L^{(I)}_\pm,\qquad [L_+^{(I)},L_-^{(I)}]=2i L_0^{(I)} \ .
\end{eqnarray}
This subalgebra sends the worldline of any null pencil in the $X^1$ direction that passes through the origin to itself.

Start with the Rindler wedge $W$, and its modular Hamiltonian $H_W= 2\pi J_{01}$ generating boosts in $X^1$ direction. The rotation by $\pi/2$ in the $(X^0,X^{d+1})$ plane maps the Rindler wedge $W$ to the causal development of the ball $B$:
\begin{eqnarray}
    &&(x')^{\mu\neq 1}=\frac{x^\mu}{x^1+\frac{1}{2}(x^2+1)}\nn\\
    &&(x')^1=\frac{(x^2-1)/2}{x^1+\frac{1}{2}(x^2+1)}
\end{eqnarray}
To see that this map sends the Rindler wedge (points with $x^1>0$ and $x^2>0$) to the interior of the ball, we compute 
\begin{eqnarray}
    1-(x')^2=\frac{4x^1}{1+x^2+2x^1}\geq 0\ .
\end{eqnarray}

Then, the modular flow of $B$ is generated by
\begin{eqnarray}
    &&H_B=e^{-i J_{1(d+1)}\pi/2}(2\pi J_{01})e^{i J_{1(d+1)}\pi/2}=-2\pi J_{0(d+1)}=\pi(P_0-K_0)\nn\\
    &&J_{1(d+1)}=\frac{1}{2}(K_1-P_1)
\end{eqnarray}
More explicitly, one can show that the modular flow of the ball is 
\begin{eqnarray}
    \alpha_s(x)^\pm=\frac{(1+x^\pm)-e^{-2\pi s}(1-x^\pm)}{(1+x^\pm)+e^{-2\pi s}(1-x^\pm)}
\end{eqnarray}
where $x^\pm=x^0\pm r$ with $r^2=x_ix_i$. It is straightforward to check that 
\begin{eqnarray}
    \Delta_B^{-i s}H_W\Delta_B^{is}=2\pi(\sinh(2\pi s)J_{1(d+1)}+\cosh(2\pi s)J_{01})
\end{eqnarray}
which implies
\begin{eqnarray}
    &&\Delta_B^{1/4}H_W \Delta_B^{-1/4}=-i2\pi J_{1(d+1)}\nn\\
    &&\Delta_B^{1/4}\Delta_W^{is}\Delta_B^{-1/4}=e^{-2\pi s J_{1(d+1)}}\\
    &&\Delta_W^{1/4}\Delta_B^{is}\Delta_W^{-1/4}=e^{2\pi s J_{1(d+1)}}
\end{eqnarray}
where $\Delta_W^{is} = e^{-is H_W}$ in our notation. This is intuitive because $\Delta_B^{-1/4}$ acts as $(X^0,X^{d+1})\to (i X^{d+1},i X^0)$.

\paragraph{Case II $(X^{-1},X^0,X^{(d+1)})$:} This subalgebra $psl(2,\mathbb{R})$ is generated by $\{J_{(-1)0},J_{0(d+1)},J_{(-1)(d+1)}\}$ with commutation relations
\begin{eqnarray}\label{caseIIalg}
        &&[J_{(-1)(d+1)},J_{(-1)0}]=-iJ_{0(d+1)}\nn\\
        &&[J_{(-1)(d+1)},J_{0(d+1)}]=-iJ_{(-1)0}\nn\\
        &&[J_{0(d+1)},J_{(-1)0}]=iJ_{(-1)(d+1)}\ .        
\end{eqnarray}
Therefore, the Lie algebra above becomes
\begin{eqnarray}
    &&[(P_0\pm K_0),D]=i(P_0\mp K_0),\qquad [P_0,K_0]=2i D\ .
\end{eqnarray}

It is convenient to define 
\begin{eqnarray}
&&L_0^{(II)}\equiv -J_{(-1)(d+1)}=-D\nn\\
&&L_+^{(II)}=-(J_{(-1)0}+ J_{0(d+1)})=P_0\nn\\
&&L_-^{(II)}=-(J_{(-1)0}- J_{0(d+1)})=K_0
\end{eqnarray}
 so that the commutation relations become
\begin{eqnarray}
    &&[L^{(II)}_0,L^{(II)}_\pm]=\pm i L^{(II)}_\pm,\qquad [L_+^{(II)},L_-^{(II)}]=-2i L_0^{(II)} \nn\\
    &&[P_0,D]=iP_0,\qquad [K_0,D]=-iK_0,\qquad [P_0,K_0]=2i D\ .
\end{eqnarray}
This subalgebra sends the worldline of an observer sitting at rest at the origin to itself.

Start with the region $B$ and perform a rotation by $\pi/2$ in the $(X^{-1},X^{0})$ plane. This maps the double cone $B$ to the forward lightcone $F$. The modular flow of $F$ is generated by 
\begin{eqnarray}
  &&  H_F=e^{-i \pi J_{(-1)0}/2}(2\pi J_{0(d+1)})e^{i \pi J_{(-1)0}/2}=-2\pi J_{(-1)(d+1)}=-2\pi D\ .
\end{eqnarray}
This time, we find
\begin{eqnarray}\label{identityCFT}
    &&\Delta_B^{1/4}\Delta_F^{is}\Delta_B^{-1/4}:e^{-2\pi s J_{(-1)0}}
\end{eqnarray}
where $J_{(-1)0}$ generates global time-translations in the universal Lorentzian cylinder. To see this, recall that $J^{(-1)0}$ generates rotations in $(X^{-1},X^0)$ closed time-like circle of $SO(d+1,2)$. To obtain the universal cover, we decompactify this circle to get the global time of the Lorentzian cylinder.

More generally, we have 
\begin{eqnarray}
    \Delta_B^{it}\Delta_F^{is}\Delta_B^{-it}:e^{2\pi i s(-\sinh(2\pi t)J_{(-1)0}+\cosh(2\pi t)J_{(-1)(d+1)})}
\end{eqnarray}
Then, it follows that 
\begin{eqnarray}
    \Delta_B^{\alpha}\Delta_F^{is}\Delta_B^{-\alpha}=e^{-2\pi s(\sin(2\pi \alpha)J_{(-1)0}-i\cos(2\pi \alpha)J_{(-1)(d+1)})}
\end{eqnarray}
This is similar to the result that we found in Rindler wedges, except that as opposed to that case, since $J_{(-0,1)}$ and $J_{(-1)(d+1)}$ do not commute, the norm of the operator above does depend on $J_{(-1)(d+1)}$ unless we set $\alpha=1/4$.

\subsection{Characteristic function of balls in CFT}
From (\ref{identityCFT}) it follows that 
\begin{eqnarray}
    &&T\equiv\Delta_B^{1/4}\Delta_F^{-is}\Delta_B^{-1/4}\Delta_F^{is}\nn\\
    &&T T^\dagger =  e^{4\pi sJ_{(-1)0}}
\end{eqnarray}
which is a Euclidean global time evolution on the Lorentzian cylinder. Since $\Delta_F^{is}=e^{- i s K_F}$ is dilatation, it sends the ball-shaped region $B$ to $B_s$ defined by its future and past tips at $X^0=\pm e^{-2\pi s}$. As a result, we find
\begin{eqnarray}
    J_{(-1)0}=\frac{1}{2\pi s}\log|\Delta_B^{1/4}\Delta_{B_s}^{-1/4}| = \frac{1}{4\pi}\p_s(TT^\dagger)\Big|_{s=0}\ .
\end{eqnarray}
The modular Hamiltonian of $B_s$ is 
\begin{eqnarray}\label{KBs}
    K_{B_s}&&=-e^{2\pi i s D}(2\pi J_{0(d+1)})e^{-2\pi i s D}= 2\pi (-\cosh(2\pi s)J_{0(d+1)}+\sinh(2\pi s)J_{(-1)0})\nn\\
    &&=\cosh(2\pi s) K_B+\sinh(2\pi s)K_H
\end{eqnarray}
where $K_B$ is the modular Hamiltonian of the unit ball, $K_H = 2\pi H$ and $H$ is the generator of the global time-evolution of the Lorentzian cylinder.
More generally, we have
\begin{eqnarray}
    \Delta_{B_s}^{-it}\frac{K_F}{2\pi}\Delta_{B_s}^{it}=-\cosh(2\pi s)\sinh(2\pi t)J_{(-1)0}-\cosh(2\pi t)J_{(-1)(d+1)}+\sinh(2\pi s)\sinh(2\pi t)J_{0(d+1)}\nn
\end{eqnarray}
which implies
\begin{eqnarray}
    \Delta_{B_s}^{\alpha}\frac{K_F}{2\pi}\Delta_{B_s}^{-\alpha}=-i\cosh(2\pi s)\sin(2\pi \alpha)J_{(-1)0}-\cos(2\pi \alpha)J_{(-1)(d+1)}+i\sinh(2\pi s)\sin(2\pi \alpha)J_{0(d+1)}\nn\\
    \ .
\end{eqnarray}

We define a family of characteristic functions $T_{u,s}(z)$ as 
\begin{align}
    T_{u,s}(z) \equiv \Delta_{B_u}^{iz}\Delta_{B_{u+s}}^{-iz},
\end{align}
such that $T = T_{0,s}(-i/4)$. Since $T_{u,s}(z) =  \Delta_F^{-iu}\Delta_{B_0}^{iz}\Delta_{B_s}^{-iz} \Delta_F^{iu}$, we first consider the case $u=0$.

We consider $PSL(2,\mathbb{R})$ and $sl_2\mathbb{C}$ in terms of $2\times2$ matrix. The generators identified as
\begin{gather}
   J_{(-1)(d+1)} \sim \begin{pmatrix} 0 & -i/2 \\-i/2 & 0\end{pmatrix},\\
    J_{(-1)0} \sim \begin{pmatrix} 0 & i/2 \\-i/2 & 0\end{pmatrix},\\
    J_{0(d+1)}  \sim \begin{pmatrix}-i/2 & 0 \\ 0 & i/2\end{pmatrix}
\end{gather}
satisfy the algebra in (\ref{caseIIalg}). 
The group elements are
\begin{gather}
    \Delta_F^{is}=e^{2\pi i s J_{(-1)(d+1)}} \sim \exp (\pi s \begin{pmatrix}0 & 1 \\1 & 0\end{pmatrix}) = \begin{pmatrix}\cosh(\pi s) & \sinh(\pi s) \\ \sinh(\pi s) & \cosh(\pi s)\end{pmatrix},\\
      e^{2\pi i s J_{(-1)0}} \sim \exp (\pi s \begin{pmatrix}0 & -1 \\1 & 0\end{pmatrix}) = \begin{pmatrix}\cos(\pi s) & -\sin(\pi s) \\ \sin(\pi s) & \cos(\pi s)\end{pmatrix},\\
    \Delta_B^{iz}= e^{2\pi iz J_{0(d+1)}} \sim \exp (\pi z \begin{pmatrix}1 & 0 \\ 0 & -1\end{pmatrix}) = \begin{pmatrix} e^{\pi z} & 0 \\ 0 & e^{-\pi z}\end{pmatrix},\\
    \Delta_{B}^{iz}\Delta_{B_s}^{-iz} \sim \begin{pmatrix} \cosh^2(\pi s) - \sinh^2(\pi s) e^{2\pi z} & \cosh(\pi s)\sinh(\pi s) (1-e^{2\pi z}) \\\cosh(\pi s)\sinh(\pi s) (1-e^{-2\pi z}) & \cosh^2(\pi s) - \sinh^2(\pi s) e^{-2\pi z}\end{pmatrix}.
\end{gather}
By diagonalization or matrix decomposition, we can find the matrix corresponding to $\log \Delta_{B}^{iz}\Delta_{B_s}^{-iz}$. The result is
\begin{align}
    \Delta_{B}^{iz}\Delta_{B_s}^{-iz}  \sim& \begin{pmatrix} 1 &1 \\u_+ & u_-\end{pmatrix}\begin{pmatrix} A+\sqrt{A^2-1} &0 \\0 & A-\sqrt{A^2-1}\end{pmatrix}\begin{pmatrix} 1 &1 \\u_+ & u_-\end{pmatrix}^{-1}\\
    =& \exp\left[\begin{pmatrix} 1 &1 \\u_+ & u_-\end{pmatrix}\begin{pmatrix} \log(A+i\sqrt{1-A^2}) &0 \\0 & \log(A-i\sqrt{1-A^2})\end{pmatrix}\begin{pmatrix} 1 &1 \\u_+ & u_-\end{pmatrix}^{-1} \right]\\
    =& \exp\left[\frac{{\log(A+i\sqrt{1-A^2})}}{u
    _+-u_-} \begin{pmatrix} -(u_+ + u_-) & 2 \\ -2u_+u_- & u_++u_-\end{pmatrix} \right],
\end{align}
where
\begin{gather}
    A = \cosh^2 (\pi s) - \sinh^2 (\pi s) \cosh(2\pi z) = 1-2\sinh^2(\pi s)\sinh^2(\pi z),\\
    u_\pm = \frac{\sinh^2(\pi s)\sinh(2\pi z)  \pm i\sqrt{1-A^2}}{\cosh(\pi s) \sinh(\pi s) (1-e^{2\pi z})}.
\end{gather}
We can write
\begin{eqnarray}
    i\cos^{-1}(A) = \log(A+i\sqrt{1-A^2})
\end{eqnarray}
for $|A| < 1$. All matrix elements are real for $z\in\mathbb{R}$. The matrix log is not unique in general. We have used the small $z$ behaviour and continuity to determine $\log$. The expression cannot be extended to $A\leq -1$. When $A\to-1^+$, the expression diverges. The expression is restricted to the range $\sinh^2(\pi s)\sinh^2(\pi z)< 1$. A Lie group element can be interpreted as an exponential of the Lie algebra element only if the group element is close to the identity.

In terms of the generators, we have 
\begin{eqnarray}
    \log \Delta_{B}^{iz}\Delta_{B_s}^{-iz} &&=\frac{2i\log(A+i\sqrt{1-A^2})}{u_+-u_-} \left( (u_++u_-)J_{0(d+1)} + (u_+u_-+1)J_{(-1)(0)} + (u_+u_--1)J_{(-1)(d+1)}\right)\nn\\
    &&=\frac{i\log(A+i\sqrt{1-A^2})}{\pi(u_+-u_-)} \left( -(u_++u_-) K_B + (u_+u_-+1)K_H - (u_+u_--1) K_F\right)\nn\\
        &&=-\beta \lb\frac{1}{\pi}\p_sK_{B(s/2)} +\tanh(\pi z)K_F\rb\nn\\
        \beta&&=-\text{sign}(z)\frac{i\cos^{-1}(A)  \cosh{(\pi z)}}{\pi\sqrt{1-\sinh^2(\pi z)\sinh^2(\pi s)}} = -\frac{\cos^{-1}(A)}{\sqrt{1-A^2}} \frac{2i}{\pi}\sinh(\pi s)\sinh(\pi z)\cosh(\pi z)
\end{eqnarray}

For $z = -i/4$, we have the previous result 
\begin{align}
    \left|\Delta_{B}^{1/4}\Delta_{B_s}^{-1/4}\right|^2 =& \Delta_{B}^{1/4}\Delta_{B_s}^{-1/4}\Delta_{B_s}^{-1/4}\Delta_{B}^{1/4} = 
     \left[\Delta_{B}^{iz}\Delta_{B_s}^{-iz}\right]|_{z=-i/4}\left[\Delta_{B}^{iz}\Delta_{B_s}^{-iz}\right]^{-1}|_{z=i/4}\\
    &=\exp\left(-2\pi is \begin{pmatrix}0 & 1 \\-1 & 0\end{pmatrix}\right) = \exp(4\pi s J_{(-1)0}).
\end{align}
Therefore, we find 
\begin{eqnarray}
    K_H=2\pi J_{(-1)0}=\frac{1}{2}\frac{d}{ds} (TT^\dagger)\Big|_{s=0}.
\end{eqnarray}

For $z = -i/2$, we have $\Delta_{B}^{1/2}\Delta_{B_s}^{-1/2} = e^{-4\pi i s J_{(-1)(d+1)}}$. 
For general $u = -\frac{1}{2\pi}\log R$, we apply dilatation to get
\begin{eqnarray}
&&T_{u,s}(-i/4)=\Delta_{B_su}^{1/4}\Delta_{B_{u+s}}^{-1/4}\nn\\
&&\tilde{H}_{u}= \frac{1}{2}\frac{d}{ds}\log(T_{u,s}(-i/4)T_{u,s}^\dagger(-i/4))\Big|_{s=0} = \Delta_F^{-iu}K_H \Delta_F^{iu}\nn\\
&&=\cosh(2\pi u)K_H+\sinh(2\pi u)K_B=\frac{1}{2\pi}\frac{d}{du}K_{B_u}=-R\p_R K_{B(R)}\ .
\end{eqnarray}
In CFT, we can interpret $\tilde{H}_{u}= \Delta_F^{-iu}H \Delta_F^{iu}$ as the global Hamiltonian of the Lorentzian cylinder transformed under dilatation.
Combined with (\ref{KBs}), we have the following equation:
\begin{eqnarray}
    \begin{pmatrix}
        & \tilde{H}_u \\
        & -K_{B_u}
    \end{pmatrix}
    =2\pi\begin{pmatrix}
        &\cosh(2\pi u) & -\sinh(2\pi u)  \\
         & -\sinh(2\pi u)& \cosh(2\pi u)
    \end{pmatrix}
    \begin{pmatrix}
        J_{(-1)0}\\
        J_{0(d+1)}
    \end{pmatrix}
\end{eqnarray}
As a result, we find
    \begin{eqnarray}
    2\pi\begin{pmatrix}
    J_{(-1)0}\\
        J_{0(d+1)}
    \end{pmatrix}
    &&=\begin{pmatrix}
        &\cosh(2\pi u) & \sinh(2\pi u)  \\
         &\sinh(2\pi u)& \cosh(2\pi u)
    \end{pmatrix}
    \begin{pmatrix}
        & \tilde{H}_{u} \\
        & -K_{B_u}
    \end{pmatrix}\nn\\
    &&=\frac{-1}{2R}\begin{pmatrix}
        &1+R^2 & 1-R^2  \\
         &1-R^2& 1+R^2
    \end{pmatrix}
    \begin{pmatrix}
        & R\p_R K_{B(R)}\\
        & K_{B(R)}
    \end{pmatrix}\ .
\end{eqnarray}
On a constant time-slice of a CFT in the vacuum on a Lorentzian cylinder we have
\begin{eqnarray}
    K_{B(R)}=\int dr d\Omega_{d-2}\: \frac{(R^2-r^2)}{2R}T_{00}(x)
\end{eqnarray}
Then, it is straightforward to check that 
\begin{eqnarray}
&&2\pi J_{(-1)0}=-\frac{1}{2}\int dr d\Omega_{d-2}\:(1+r^2) T_{00}(x)=\int_{0}^\pi d\theta d\Omega_{d-2}\: T_{00}(\theta)
\end{eqnarray}
which is the generator of global time translations on the Lorentzian cylinder, because the embedding of Miknowski coordinates into the Lorentzian cylinder is given by the map
\begin{eqnarray}
    \tan\lb\frac{\tau\pm \theta}{2}\rb=t\pm r
\end{eqnarray}
which maps $t=0$ surface to $\tau=0$ with $\tan(\theta/2)=r$.

Finally, we consider the characteristic function for general $B(u)$ and $B(u+s)$. We have
\begin{eqnarray}
   \log T_{u,s}(z)=\log\Delta_{B_u}^{iz}\Delta_{B_{u+s}}^{-iz}=\Delta_F^{-iu} \log T_{0,s}(z) \Delta_F^{iu} 
\end{eqnarray}
Under the flow $\Delta_F^{-iu}$, this gives
\begin{eqnarray}
 &&\begin{pmatrix}
        & J_{(-1)0} \\
        & J_{0(d+1)}
    \end{pmatrix}
    \to \begin{pmatrix}
        &\cosh(2\pi u) & -\sinh(2\pi u)  \\
         & -\sinh(2\pi u)& \cosh(2\pi u)
    \end{pmatrix}
    \begin{pmatrix}
        J_{(-1)0}\\
        J_{0(d+1)}
    \end{pmatrix}\nn\\
    &&K_B\to  K_{B_u}=K_{B(R)}, \qquad K_H\to \tilde{H}_{u} =-R\p_R K_{B(R)}, \qquad K_F\to K_F.
\end{eqnarray}
Therefore, we have the following result
\begin{eqnarray}
    \log T_{0,s}&&=-\beta \lb\frac{1}{\pi}\p_sK_{B_{s/2}} +\tanh(\pi z)K_F\rb\nn\\
    \log T_{u,s}&&=-\beta \lb \frac{1}{\pi}(\p_s K_{B_{u+s/2}})+\tanh(\pi z)K_F\rb \nn\\
    &&= -\beta \lb \frac{1}{2\pi}(\p_u K_{B_{u+s/2}})+\tanh(\pi z)K_F\rb
\end{eqnarray}

We can also construct the global Hamiltonian $J_{(-1)0}$ through $T_{0,s}(z)$ with real $z$. We consider the limit $\pi s \ll 1$ and $ \sinh^2(\pi s)\sinh^2(\pi z)\ll 1$. We have
\begin{align}
    \log (\Delta_{B}^{iz}\Delta_{B_s}^{-iz}) =& \begin{pmatrix} 0& \pi s(1-e^{2\pi z}) \\ \pi s(1-e^{-2\pi z}) & 0 \end{pmatrix} + O((\pi s)^2) \\
    = &\: i 2 \pi s \sinh(2\pi z) J_{(-1)0} - i 4\pi s \sinh^2(\pi z)J_{(-1)(d+1)} + O((\pi s)^2).
\end{align}
By comparing the forward and backward flow, we have
\begin{align}
    \log (\Delta_{B}^{iz}\Delta_{B_s}^{-iz}) - \log (\Delta_{B}^{-iz}\Delta_{B_s}^{iz}) = i 4 \pi s \sinh(2\pi z) J_{(-1)0} + O((\pi s)^2).
\end{align}
Notice that it is not the same as $\log(T_{0,s}(z)(T_{0,s}(-z))^{-1}$. The two terms in the $\log$ function are not commuting. In fact, we can compute 
\begin{align}
    T_{0,s}(z)(T_{0,s}(-z))^{-1} &= \begin{pmatrix}\cosh^2(\pi s) - e^{4\pi z} \sinh^2(\pi s) & -\sinh(2\pi s)\sinh(2 \pi z) \\ \sinh(2\pi s)\sinh(2 \pi z) & \cosh^2(\pi s) - e^{-4\pi z} \sinh^2(\pi s) 
    \end{pmatrix} \\
    &=I - \sinh(2\pi s)\sinh(2\pi z)\begin{pmatrix}e^{2\pi z} \tanh(\pi s) & 1 \\-1 & -e^{-2\pi z}\tanh(\pi s)\end{pmatrix}.
\end{align}
The expression is valid in the sense of $2\times2$ matrix. We  have
\begin{align}
    &\log T_{0,s}(z)T_{0,s}(-z)^{-1}= i\kappa(\sinh(\pi s) \cosh(2\pi z) K_B + \cosh(\pi s) K_H),\\
    &\kappa = \frac{\log(C+i\sqrt{1-C^2})}{2\pi i\sqrt{1-C^2}} 4\sinh(\pi s) \sinh(2\pi z),\\
    &C = 1 -2 \sinh^2(\pi s)\sinh^2(2\pi z).
\end{align}
For small $s$ limit, we have
\begin{gather}
    \log (\Delta_{B}^{iz}\Delta_{B_s}^{-i2z}\Delta_{B}^{iz})= i 4 \pi s \sinh(2\pi z) J_{(-1)0} + O((\pi s)^2),\\
    \log (\Delta_{B}^{iz}\Delta_{B_s}^{-i2z}\Delta_{B}^{iz}) \approx \log (\Delta_{B}^{iz}\Delta_{B_s}^{-iz}) - \log (\Delta_{B}^{-iz}\Delta_{B_s}^{iz}) \approx i 4 \pi s \sinh(2\pi z) J_{(-1)0}.
\end{gather}
The two expressions are the same up to the first order of $\pi s$. If we consider the analytic continuation of $\log (\Delta_{B}^{iz}\Delta_{B_s}^{-i2z}\Delta_{B}^{iz})$ at $z = -i/4$, this gives $\log |T_{0,s}(-i/4)|^2 = -4 \pi s J_{(-1)0}$, which is exact without any higher order terms. 
\begin{lemma}[Approximation of the Global Hamiltonian]
    Consider a CFT in Minkowski space. For $B$ being the unit ball centered at the origin and $B_s$ being the ball centered at the origin with radius $e^{-2\pi s}$, the global Hamiltonian $J_{(-1)0}$ can be approximated by $T_{0,s}(z)(T_{0,s}(-z))^{-1}$ under the limit $0<\pi s \ll 1$, $ \sinh^2(\pi s)\sinh^2(\pi z)\ll 1$, and $z\in\mathbb{R}$. We can express it as 
    \begin{gather}
        \log (\Delta_{B}^{iz}\Delta_{B_s}^{-i2z}\Delta_{B}^{iz}) = i 4 \pi s \sinh(2\pi z) J_{(-1)0} + O((\pi s)^2),\\
        \log (\Delta_{B}^{iz}\Delta_{B_s}^{-iz}) - \log (\Delta_{B}^{-iz}\Delta_{B_s}^{iz}) = i 4 \pi s \sinh(2\pi z) J_{(-1)0} + O((\pi s)^2).
    \end{gather}
\end{lemma}
For general real $u$, we find $T_{u,s}(z)(T_{u,s}(-z))^{-1}$ as the following
\begin{eqnarray}
    &&\frac{1}{\sinh(2\pi s)\sinh(2\pi z)}(I - T_{u,s}(z)(T_{u,s}(-z))^{-1}) =\begin{pmatrix}
        C_{11}& C_{12}\\
        C_{21}&C_{22}
    \end{pmatrix}\nn\\
    &&C_{11}=\tanh(\pi s)(e^{2\pi z}\cosh^2(\pi u) + e^{-2\pi z}\sinh^2(\pi u)) + \sinh(2\pi u) \nn\\
    &&C_{12}=\tanh(\pi s)\cosh(2\pi z)\sinh(2\pi u) + \cosh(2\pi u) 
    \nn\\
    &&C_{21}=-\tanh(\pi s)\cosh(2\pi z)\sinh(2\pi u) - \cosh(2\pi u) \nn\\
    &&C_{22}= -\tanh(\pi s)(e^{-2\pi z}\cosh^2(\pi u) + e^{2\pi z}\sinh^2(\pi u)) - \sinh(2\pi u)\ .
\end{eqnarray}
The generator is 
\begin{align}
    \log T_{u,s}(z)(T_{u,s}(-z))^{-1} = i\kappa(\sinh(\pi s) \cosh(2\pi z) K_{B_u} + \cosh(\pi s) \tilde{H}_{u}).
\end{align}

\section{Kubo-Ando mean}

\begin{corollary}
For the setup of Theorem 1, we have
\begin{enumerate}
    \item  $T_a T_a^\dagger=e^{-2a\sin(\alpha) H}$.
    \item $(T_a T_a^\dagger)^\gamma=T_{a\gamma} T_{a\gamma}^\dagger$.
    \item $\Delta_A^{-1/4}\sharp_\gamma \Delta_B^{-1/4}=\Delta_{W(a\gamma)}^{-1/2}$ 
\end{enumerate}
\end{corollary}
where for positive operators $X$ and $Y$ and $\gamma\in (0,1)$, the Kubo-Ando geometric mean is defined to be 
\begin{eqnarray}
    X\sharp_\gamma Y=X^{1/2}\lb X^{-1/2}YX^{-1/2} \rb^\gamma X^{1/2}\ .
\end{eqnarray}
A key property of the Kubo-Ando mean is its monotonicity
\begin{eqnarray}
    \forall X<X',Y<Y':\qquad (X\sharp_\gamma Y)<(X'\sharp_\gamma Y')\ .
\end{eqnarray}
One may wonder whether the monotinicity condition above results in new inequalities in QFT.
In our case of inclusion of wedges in QFT these monotonicity equations do not result in new equations,  because by virtue of the third part of the corollary above, the Kubo-Ando mean of modular operators becomes the modular of another region. A similar phenomenon occurs in the vacuum of conformal field theory.

\section{Perturbations of $\tilde{H}$ \label{app:Ht}}
\begin{lemma}\label{Htilde}
Let $K(s)$ be a one-parameter family of operators on a Hilbert space that is differentiable at $s=0$. Let $z \in \mathbb{C}$ be a complex parameter. Consider the operators defined by
\begin{eqnarray}
    &&S(z,s)=e^{-iz K(0)} e^{2iz K(s)} e^{-iz K(0)}\ .
\end{eqnarray}
The first variation with respect to $s$ is given by
    \begin{eqnarray}
        \tilde{H}(z) &&=\frac{{-}1}{2\sin(2\pi iz)}\p_s\log S(z,s)\Big|_{s=0}= \frac{{-}1}{2\sinh(2\pi z)}\int_{-z}^{z} dv \, e^{i v \ad_{K(0)}} \dot{K}(0)\nn\\
        &&=\frac{{-}\sin(z \ad_{K(0)})}{\sinh(2\pi z)\ad_{K(0)}} \dot{K}(0)
    \end{eqnarray}
    where $\dot{K}(0) = \p_s K(s)|_{s=0}$ and the adjoint action is defined as $\ad_A B = [A, B]$. In the limit of $z\to 0$ we have
    \begin{eqnarray}
        2\pi\tilde{H}(0)={-}\dot{K}(0)\ .
    \end{eqnarray}
\end{lemma}

\begin{proof}
We linearize the operator $K(s)$ around $s=0$. To first order in small $s$, we have $K(s) \approx K(0) + s \dot{K}(0)$.
We define the parameter $\alpha = iz$ and apply Duhamel's formula for the derivative of the exponential map:
\begin{equation}
    e^{2iz K(s)} \approx e^{2iz K(0)} + 2iz s \int_0^1 du \, e^{2iz(1-u)K(0)} \dot{K}(0) e^{2iz u K(0)}
\end{equation}
Substituting this into $S(z,s)$ and simplifying the sandwiching terms:
\begin{align}
    S(z, s) &\approx 1 + 2iz s \int_{1}^{-1} \left(-\frac{dv}{2}\right) e^{v (iz) K(0)} \dot{K}(0) e^{-v (iz) K(0)} \\
      &= 1 + iz s \int_{-1}^{1} dv \, e^{v (iz) \ad_{K(0)}} \dot{K}(0)
\end{align}
where we changed variables to $v=1-2u$. We evaluate the integral over $v$:
\begin{equation}
    \int_{-1}^{1} dv \, e^{v \lambda} = \frac{e^{\lambda} - e^{-\lambda}}{\lambda} = \frac{2\sinh(\lambda)}{\lambda}
\end{equation}
Setting $\lambda = iz \ad_{K(0)}$, we find
\begin{eqnarray}
    S(z, s) &&\approx 1 + i s \left( 2\frac{\sin(z \ad_{K(0)})}{ \ad_{K(0)}} \dot{K}(0) \right)\nn\\
    \tilde{H}(z) &&= \frac{{-}1}{2\sin(2\pi iz)} \p_s S(z, s)\Big|_{s=0} \nn \\
    &&= \frac{{-}\sin(z \ad_{K(0)})}{\sinh(2\pi z) \ad_{K(0)}} \dot{K}(0)
\end{eqnarray}
\end{proof}
These expression should be understood using complex interpolation theory of non-commutative $L_p$-spaces \cite{furuya2023monotonic}.
\begin{lemma}\label{lemmaanitocmmuteApp}
For $\beta \in \mathbb{R}$, and $\delta K$ the first order perturbation of self-adjoint operator $K$, we have the identity
\begin{equation}
    \frac{1}{2\beta} \int_{-\beta}^\beta d\tau \, e^{\tau K} (\delta K) e^{-\tau K} = \frac{-1}{2\beta} \left\{ e^{\beta K}, \delta(e^{-\beta K}) \right\}
\end{equation}
where $\{A, B\} = AB + BA$ denotes the anti-commutator.
\end{lemma}

\begin{proof}
We consider the variation of the operator $e^{2\beta K}$. Using the chain rule for the variation of a squared operator $A^2$ (where $A=e^{\beta K}$):
\begin{equation}
    \delta(e^{2\beta K}) = \delta \left( e^{\beta K} e^{\beta K} \right) = \delta(e^{\beta K}) e^{\beta K} + e^{\beta K} \delta(e^{\beta K})
\end{equation}
Alternatively, using Duhamel's formula for the variation of an exponential, we have:
\begin{equation}
    \delta(e^{2\beta K}) = \int_0^1 ds \, e^{(1-s)2\beta K} (2\beta \delta K) e^{s 2\beta K}
\end{equation}
We conjugate this expression by $e^{-\beta K}$:
\begin{align}
    e^{-\beta K} \delta(e^{2\beta K}) e^{-\beta K} &= e^{-\beta K} \left( 2\beta \int_0^1 ds \, e^{2\beta K} e^{-2s \beta K} (\delta K) e^{2s \beta K} \right) e^{-\beta K} \nonumber \\
    &= 2\beta \int_0^1 ds \, e^{\beta K(1-2s)} (\delta K) e^{-\beta K(1-2s)}
\end{align}
We perform the change of variables $\tau = \beta(1-2s)$, with $d\tau = -2\beta ds$. The limits change from $[0, 1]$ to $[\beta, -\beta]$:
\begin{equation}
    e^{-\beta K} \delta(e^{2\beta K}) e^{-\beta K} = \int_{-\beta}^{\beta} d\tau \, e^{\tau K} (\delta K) e^{-\tau K}
\end{equation}
Substituting the product rule expansion into the left-hand side:
\begin{align}
    e^{-\beta K} \left( \delta(e^{\beta K}) e^{\beta K} + e^{\beta K} \delta(e^{\beta K}) \right) e^{-\beta K} &= e^{-\beta K} \delta(e^{\beta K}) + \delta(e^{\beta K}) e^{-\beta K} \nonumber \\
    &= \{ e^{-\beta K}, \delta(e^{\beta K}) \}
\end{align}
Dividing both sides by $2\beta$ completes the proof.
\end{proof}
\begin{corollary}
Take $K(s)=K_{B(R e^{-2\pi s})}$ to be the Hamiltonian of a ball-shaped region of radius $R e^{-2\pi s}$, then
\begin{eqnarray}
2\pi \tilde{H}(0)={-}\dot{K}(0)=2\pi R\p_R K_{B(R)}\ .
\end{eqnarray}
As we vary $z$ in the complex plane, the expression $\tilde{H}(z)$ for real $z$ can be viewed as a channel applied to $-R \p_R K_{B(R)}$ that averages over modular time in the range $(-z,z)$:
\begin{eqnarray}
    &&\tilde{H}(z)=\mathcal{E}_z(R \p_R K_{B(R)})\nn\\
    &&\mathcal{E}_z(X)=\frac{1}{2\sinh(2\pi z)} \int_{-z}^z du \: e^{i u K} X e^{-i u K} \nn\\
    &&\tilde{H}(z)=\frac{i}{2\sinh(2\pi z)}\{\Delta^{iz},\p_s \Delta^{-iz}\}
\end{eqnarray}
and for imaginary values of $z=-i\beta/2$ it can be interpreted as applying an average over Euclidean modular flow:
\begin{eqnarray}
    &&\mathcal{E}_{-i\beta/2}(X)=\frac{1}{2\sin(\beta/2)}\int_{-\beta/2}^{\beta/2}d\tau e^{\tau K} X e^{-\tau K}\nn\\
    &&\tilde{H}(-i\beta/2)=\frac{1}{2\sin(\beta/2)}\{\Delta^{-\beta/2},\p_s \Delta^{\beta/2}\}\ .
\end{eqnarray}
The path-integral that computes $\Delta^\alpha$ is a Euclidean double cone cut open. Our Hamiltonians compute the response of this Euclidean double cone to the regulator $\p_s$ that separates the two sides of the cone by moving them in opposite directions.
    
\end{corollary}

\begin{lemma}\label{lemmaperturbApp}
Let $K(s, \lambda)$ be a two-parameter family of operators on a Hilbert space, smooth in both $s$ and $\lambda$. Let $z \in \mathbb{C}$. 
The effective Hamiltonian $\tilde{H}(z, \lambda)$, defined by
\begin{equation}
    \tilde{H}(z, \lambda) = \frac{{-}1}{2\sin(2\pi iz)} \p_s \log \left( e^{-iz K(0, \lambda)} e^{2iz K(s, \lambda)} e^{-iz K(0, \lambda)} \right) \Big|_{s=0},
\end{equation}
admits the following expansion to first order in $\lambda$:
\begin{equation}
    \tilde{H}(z, \lambda) = \tilde{H}(z,0) + \lambda \tilde{H}'(z,0) + \mathcal{O}(\lambda^2)
\end{equation}
where 
\begin{align}
    \tilde{H}(z,0) &= \frac{{-}\sin(z \ad_K)}{\sinh(2\pi z) \ad_K} \dot{K}(0,0) \\
    \tilde{H}'(z,0) &= \frac{{-}\sin(z \ad_K)}{\sinh(2\pi z) \ad_K} (\dot{K}') {-}\frac{iz^2}{2\sinh(2\pi z)} \int_{-1}^{1} dy \, y \int_0^1 dw \, e^{izy(1-w)\ad_K} \Big[ K', \, e^{izyw\ad_K} \dot{K} \Big]
\end{align}
with $K = K(0,0)$, $\dot{K} = \p_s K$, $K' = \p_\lambda K$, $\dot{K}' = \p_\lambda \p_s K$. 
\end{lemma}

\begin{proof}
The term $\tilde{H}(z,0)$ and the term involving $\dot{K}'$ follow directly from Lemma \ref{Htilde}. To derive the second term in (\ref{Hprimez0}). 
Starting from the integral representation derived previously:
\begin{equation}
    \frac{\sin(z \ad_K)}{z \ad_K} = \frac{1}{2} \int_{-1}^1 dy \, e^{iy z \ad_K}=\frac{1}{2z}\int_{-z}^z d\tau e^{i\tau \ad_K}
\end{equation}
This is a low-pass filter. 
We compute the variation with respect to $K$ in the direction $K'$ using Duhamel's formula:
\begin{equation}
    \delta \left( e^{X} \right) = \int_0^1 dw \, e^{(1-w)X} (\delta X) e^{wX}
\end{equation}
Here $X = iy z \ad_K$ and $\delta X = iy z \ad_{K'}$. Substituting these into the integral over $y$:
\begin{align}
    \delta \left( \frac{\sin(z \ad_K)}{z \ad_K} \right) \dot{K} &= \frac{1}{2} \int_{-1}^1 dy \, \int_0^1 dw \, e^{iyz(1-w)\ad_K} (iyz \ad_{K'}) e^{iyzw\ad_K} \dot{K} \\
    &= \frac{iz}{2} \int_{-1}^1 dy \, y \int_0^1 dw \, e^{iyz(1-w)\ad_K} \left[ K', e^{iyzw\ad_K} \dot{K} \right]\\
    &= \int_{-1}^1 dy \, i\sinh(\pi y z) [e^{iyz\ad_K/2} \mathcal{E}_{yz/2}(K'), e^{iyz\ad_K} \dot{K}]
\end{align}
{We have used $e^{i u K} X e^{-i u K} = \Ad_{e^{i u K}} X = e^{i u \ad_K} X$}.
This completes the proof. Of course, in the $z\to 0$, we recover 
\begin{eqnarray}
    \tilde{H}(z,\lambda)=\dot{K}(0,0)+\dot{K'}(0,0)+O(\lambda^2)\ .
\end{eqnarray}
\end{proof}
```````````````````````````
\bibliographystyle{unsrt}``````````
\bibliography{PRLversionNotes}


\end{document}